\pdfoutput=1
\documentclass[11pt,a4paper]{article}

\usepackage[acceptedWithA]{tacl2021v1}

\usepackage{url}
\usepackage[T1]{fontenc}

\usepackage{float}

\usepackage{microtype}
\usepackage{graphicx} %
\usepackage[graphicx]{realboxes}

\usepackage{booktabs}

\usepackage{xspace}

\usepackage[inline]{enumitem}
\newlist{compactitem}{itemize}{3}
\setlist[compactitem]{topsep=0pt,partopsep=0pt,itemsep=0pt,parsep=0pt,leftmargin=1em}
\setlist[compactitem,1]{label=\textbullet}
\setlist[compactitem,2]{label=---}
\setlist[compactitem,3]{label=*}
\newlist{compactdesc}{description}{3}
\setlist[compactdesc]{topsep=0pt,partopsep=0pt,itemsep=0pt,parsep=0pt,leftmargin=1em}
\newlist{compactenum}{enumerate}{3}
\setlist[compactenum]{topsep=0pt,partopsep=0pt,itemsep=0pt,parsep=0pt,leftmargin=1.5em}
\setlist[compactenum,1]{label=\arabic*.}
\setlist[compactenum,2]{label=(\alph*)}
\setlist[compactenum,3]{label=\roman*.}

\usepackage{url}          %
\usepackage{xcolor}
\usepackage{amsmath}
\usepackage{blkarray} \usepackage{bigstrut}

\usepackage{mathtools}
\usepackage{amssymb}
\usepackage{amsfonts}
\usepackage{amstext}
\usepackage{amsthm}
\usepackage{latexsym}
\usepackage{stmaryrd}

\allowdisplaybreaks[4]

\usepackage[capitalise,nameinlink]{cleveref} %
\crefformat{section}{#2\S #1#3}
\Crefformat{section}{#2\S #1#3}
\crefname{theorem}{Thm.}{Thms.}
\crefname{lemma}{Lem.}{Lems.}
\crefname{corollary}{Cor.}{Cors.}
\crefname{definition}{Def.}{Defs.}
\crefname{example}{Ex.}{Examples}
\crefname{proposition}{Prop.}{Props.}
\crefname{equation}{Eq.}{Eqs.}

\usepackage{comment}

\usepackage[scale=0.98]{inconsolata}
\usepackage[nott]{newtx}

\newtheorem{theorem}{Theorem}[section]
\newtheorem{lemma}[theorem]{Lemma}
\newtheorem{corollary}[theorem]{Corollary}
\newtheorem{proposition}[theorem]{Proposition}
\newtheorem{conjecture}[theorem]{Conjecture}
\theoremstyle{definition}
\newtheorem{definition}[theorem]{Definition}
\newtheorem{example}[theorem]{Example}

\usepackage{tikz}
\usetikzlibrary{arrows,arrows.meta,positioning,shapes,decorations,automata,backgrounds,petri,calc}
          \tikzset{
          initial text={}, 
          every edge/.style={draw,->},
          >=stealth
          }

\usepackage{relsize}
\tikzset{fontscale/.style = {font=\relsize{#1}}}

\newcommand{\marker}{\thinspace\text{|}\thinspace}

\newcommand{\FO}{aperiodic}
\newcommand{\FOinit}{Aperiodic}

\newcommand{\transname}{\textit}
\newcommand{\rotateright}{\transname{rotate-right}}

\newcommand{\mapduplicate}{\transname{map-duplicate}}
\newcommand{\mapreverse}{\transname{map-reverse}}

\newcommand{\majorityrules}{\transname{majority-rules}}

\newcommand{\markedsquare}{\transname{marked-square}}

\newcommand{\copyfirsthalf}{\transname{copy-first-half}}

\newcommand{\increment}{\transname{increment}}
\newcommand{\residuesmodm}{\transname{residues-mod-$m$}}
\newcommand{\modularcountermodm}{\transname{count-mod-$m$}}

\newcommand{\BRASP}{\ensuremath{\text{B-RASP}}}
\newcommand{\BRASPpos}{\ensuremath{\text{B-RASP}[\vecname{pos}]}}
\newcommand{\SRASP}{\ensuremath{\text{S-RASP}}}

\DeclareMathOperator{\leftmost}{{\blacktriangleleft}}
\DeclareMathOperator{\rightmost}{{\blacktriangleright}}
\newcommand{\att}[5]{\ensuremath{{#1}_{#2} \left[#3, #4\right] \; #5}}
\newcommand{\attl}[4]{\att{\leftmost}{#1}{#2}{#3}{#4}}
\newcommand{\attr}[4]{\att{\rightmost}{#1}{#2}{#3}{#4}}

\newcommand{\attldefault}[5]{\att{\leftmost}{#1}{#2}{#3}{#4}:{#5}}
\newcommand{\attrdefault}[5]{\att{\rightmost}{#1}{#2}{#3}{#4}:{#5}}

\newcommand{\attshort}[4]{\ensuremath{{#1}_{#2} \left[#3\right] \; #4}}
\DeclareMathOperator{\sumsymb}{{\bf psum}}
\newcommand{\attsum}[3]{\attshort{\sumsymb}{#1}{#2}{#3}}

\newcommand{\rechar}[1]{\text{`{#1}'}}
\newcommand{\charepsilon}{\rechar{$\varepsilon$}}
\newcommand{\pad}[0]{\text{\#}}
\newcommand{\iftext}{\mathrel{\text{if}}}

\newcommand{\elsetext}{\mathrel{\text{else}}}

\newcommand{\R}{\mathbb{R}}

\newcommand{\TC}{\mathsf{TC}}
\newcommand{\AC}{\mathsf{AC}}

\newcommand{\true}{\top}
\newcommand{\false}{\bot}
\newcommand{\vecname}[1]{\ensuremath{\mathtt{#1}}}
\newcommand{\pos}{\vecname{pos}}
\newcommand{\inputx}{\vecname{in}}
\newcommand{\outputy}{\vecname{out}}

\newcommand{\progname}[1]{\mathcal{#1}}

\newif\ifcomments
\commentstrue
\ifcomments
\newcommand{\LS}[1]{({\textcolor{magenta}{\textbf{LS: #1}}})}
\newcommand{\AS}[1]{({\textcolor{blue}{\textbf{AS: #1}}})}
\newcommand{\DA}[1]{({\textcolor{teal}{\textbf{DA: #1}}})}
\newcommand{\JR}[1]{({\textcolor{brown}{\textbf{JR: #1}}})}
\newcommand{\DC}[1]{({\textcolor{cyan}{\textbf{DC: #1}}})}

\newcommand{\REditor}[1]{({\textcolor{red}{AC: #1}})}
\newcommand{\RA}[1]{({\textcolor{red}{ReviewerA: #1}})}
\newcommand{\RB}[1]{({\textcolor{red}{ReviewerB: #1}})}
\newcommand{\RC}[1]{({\textcolor{red}{ReviewerC: #1}})}

\newcommand{\Instruction}[1]{{\textcolor{red}{#1}}}
\else
\newcommand{\LS}[1]{}
\newcommand{\AS}[1]{}
\newcommand{\DA}[1]{}
\newcommand{\JR}[1]{}
\newcommand{\DC}[1]{}
\newcommand{\Instruction}[1]{}
\newcommand{\REditor}[1]{}
\newcommand{\RA}[1]{}
\newcommand{\RB}[1]{}
\newcommand{\RC}[1]{}
\fi

\newcommand{\eos}{{\mathord\dashv}}
\newcommand{\rev}{\text{R}}
\newcommand{\word}{string}
\newcommand{\words}{strings}

\newcommand{\exampletablefontsize}{\footnotesize}
\newenvironment{raspcode}{\par\vspace*{-2ex}\begin{small}}{\end{small}\ignorespacesafterend}

\newcommand{\bnfto}{\mathrel{::=}}
\newcommand{\bnfmid}{\mathrel{\makebox[\widthof{$\bnfto$}][r]{$\mid$}}}

\newcommand{\bool}{\textsf{bool}}
\newcommand{\chr}{\textsf{char}}
\newcommand{\nat}{\textsf{nat}}

\newcommand{\fv}{\textnormal{FV}}

\newcommand{\Compressed}{Packed}
\newcommand{\compressed}{packed}

\newcommand{\conc}{\cdot}

\newcommand{\direction}{choice function}

\title{Transformers as Transducers}
    
\author{Lena Strobl\\
  Umeå University\\
  \href{mailto:lena.strobl@umu.se}{\tt lena.strobl@umu.se} \\\And
  Dana Angluin\\
  Yale University\\
  \href{mailto:dana.angluin@yale.edu}{\tt dana.angluin@yale.edu}\\\And
  David Chiang\\
  University of Notre Dame\\
  \href{mailto:dchiang@nd.edu}{\tt dchiang@nd.edu} \\\AND
  Jonathan Rawski\\
  San Jos{\'e} State University\\  
  \href{mailto:jon.rawski@sjsu.edu}{\tt jon.rawski@sjsu.edu}\\\And
  Ashish Sabharwal\\
  Allen Institute for AI\\
  \href{mailto:ashishs@allenai.org}{\tt ashishs@allenai.org} \\}

\begin{document}

\maketitle

\begin{abstract}
    We study the sequence-to-sequence mapping capacity of transformers by relating them to finite transducers, and find that they can express surprisingly large classes of (total functional) transductions.
    We do so using variants of RASP, a programming language designed to help people ``think like transformers,'' as an intermediate representation.
    We extend the existing Boolean variant $\BRASP$ to sequence-to-sequence transductions and show that it computes exactly the first-order rational transductions (such as string rotation).
    Then, we introduce two new extensions. $\BRASPpos$ enables calculations on positions (such as copying the first half of a string) and contains all first-order regular transductions. 
    $\SRASP$ adds prefix sum, which enables additional arithmetic operations (such as squaring a string) and contains all first-order polyregular transductions.
    Finally, we show that masked average-hard attention transformers can simulate $\SRASP$.

\end{abstract}

\section{Introduction}

Transformers \citep{vaswani-etal-2017-attention} have become a standard tool in natural language processing and vision tasks.
They are primarily studied in terms of their expressivity (which functions they can or cannot compute) or learnability (which functions they can or cannot learn from examples). Much recent expressivity work views transformers as recognizers of formal languages, by comparing them to automata, circuits, or logic \citep{strobl2023transformers}. Here we take the more general view that they compute (total functional) transductions, or functions from \words~to \words.

Transductions are a fundamental object in computer science, with a long history in linguistics and natural language processing \cite{mohri1997finite,roark2007computational}.
Many empirical tests of transformer reasoning ability use transductions to define algorithmic sequence generation tasks %
\citep[e.g.,][]{suzgun2022challenging,deletang-etal-2023-chomsky}
such as %
tracking shuffled objects, sorting strings, concatenating all $k$-th letters, or removing duplicates from a list.

This paper is the first theoretical analysis, to our knowledge, of transformers as transducers of formal languages (\cref{fig:overview}). 
Previous work on transformers as recognizers showed that unique-hard attention transformers correspond to star-free regular languages \citep{yang2024masked};
here, we prove the analogous result for transformers as transducers, that unique-hard attention transformers correspond to \emph{aperiodic rational transductions}.
We then study two superclasses of aperiodic rational transductions that are (also) analogous to star-free regular languages: 
\emph{\FO{} regular transductions} (e.g., $w \mapsto w^\rev$ or $w \mapsto ww$) and \emph{\FO{} polyregular transductions} (e.g., $w \mapsto w^{|w|}$). We prove unique-hard attention transformers cannot compute all of these, but average-hard attention transformers can.

To do this, we introduce two new variants of RASP \citep{weiss-etal-2021-rasp}, a programming language designed to make it easier to write down the kinds of computations that transformers can perform.
This makes our analysis more simple, concise, and interpretable compared to describing transformers directly using linear algebra.
These variants, called $\BRASPpos$ and $\SRASP$, compute more than just the \FO{} regular and \FO{} polyregular transductions, and are interesting in their own right.

\begin{figure}[t!]
\centering
\resizebox{\linewidth}{!}{%
\begin{tikzpicture}[x=3.5cm,y=3cm]
\newcommand{\lang}[1]{\textcolor{gray}{\footnotesize #1}}
\crefname{theorem}{Thm.}{Thms.}
\crefname{proposition}{Prop.}{Props.}
\tikzset{every node/.style={align=center}}
\tikzset{every edge/.append style={thick,>={Stealth[length=2mm,width=1.5mm]},every node/.style={sloped,font={\footnotesize}}}}
\tikzset{previous/.style={dashed}}

\node(arational) at (1,0) {\FO{}\\rational \\ \lang{$\increment$} \\[-4pt] \lang{$\rotateright$}};
\node(aregular) at (1,1) {\FO{}\\regular \\ \lang{$\mapreverse$} \\[-4pt] \lang{$\mapduplicate$}};
\node(apolyreg) at (1,2) {\FO{}\\polyregular \\ \lang{$\markedsquare$}};

\begin{scope}[xshift=-10pt]
\node(brasp) at (2.1,0) {$\BRASP$};
\node(brasppos) at (2.1,1) {$\BRASPpos$ \\ \lang{$\copyfirsthalf$}};
\node(SRASP) at (2.1,2) {$\SRASP$ \\ \lang{$\majorityrules$}};
\end{scope}
 
\node(uhat) at (3,0) {unique \\ hard attention \\ transformers \\ (no pos)};
\node(ahat) at (3,2) {average \\ hard attention \\ transformers};

\draw (arational) edge[<->] node[above] {\cref{theorem:BRASP-is-FOTop}} node[below] {\cref{thm:arational_to_brasp}} (brasp);
\draw (brasp) edge[<->,previous] (uhat);
\draw (aregular) edge node[below] {\cref{thm:aregular_to_brasppos}} (brasppos);
\draw[transform canvas={yshift=3pt}] (apolyreg) edge node[above] {\cref{theorem:a-rasp-includes-fo-polyregular}} (SRASP);
\draw (brasppos) edge node {$\not$} node[below,sloped] {\cref{thm:copyfirsthalf}} (apolyreg);
\draw[transform canvas={yshift=-3pt}] (SRASP) edge node {$\not$} node[below] {\cref{thm:majority-rules-not-in-polyregular}} (apolyreg); %
\draw (SRASP) edge node[above] {\cref{thm:srasp_to_ahat}} (ahat);

\draw (arational) edge[previous] (aregular);
\draw (aregular) edge[previous] (apolyreg);

\draw (brasp) edge (brasppos);
\draw (brasppos) edge (SRASP);

\draw (uhat) edge node[above] {\cref{thm:uhat_to_ahat}} (ahat);
\end{tikzpicture}%
}
\vspace{-3ex}
\caption{Overview of results of this paper. Arrows denote inclusion; dashed arrows denote inclusions that are known from previous work. Slashed arrows denote non-inclusions. The columns, from left to right, are: %
 (1) the hierarchy of \FO{} transductions; (2) RASP variants; (3) variants of transformer encoders.}
\label{fig:overview}
\end{figure}

\section{Preliminaries}
\label{sec:transductionandtransducers}

We write $[n]$ for the set $\{0, \ldots, n-1\}$.
Fix finite input and output alphabets $\Sigma$ and $\Gamma$. We sometimes use special symbols~$\pad$ and $\eos$, which we assume do not belong to $\Sigma$ or $\Gamma$.
Let $\Sigma^*$ and $\Gamma^*$ be the sets of strings over $\Sigma$ and $\Gamma$, respectively.
The empty string is denoted $\varepsilon$.
For any string $w$, we number its positions starting from $0$, so $w = a_0 \cdots a_{n-1}$.
We write $uv$ or $u \conc v$ for the concatenation of strings $u$ and $v$, and $w^\rev$ for the reverse of string $w$.

\subsection{Transductions and Transducers}

A transduction is a binary relation between strings in $\Sigma^*$ and strings in $\Gamma^*$. Here
we consider only total functional transductions, that is, functions $\Sigma^* \rightarrow \Gamma^*$, and all of our transducers define total functional transductions.

\begin{definition}[string homomorphism]
\label{def:string-hom}
    A \emph{string homomorphism} is a function $f \colon \Sigma^* \to \Gamma^*$ such that, for any strings $u, v \in \Sigma^*$, we have $f(uv) = f(u)\,f(v)$. %
\end{definition}

\begin{definition}[deterministic finite transducer]
A \emph{deterministic finite transducer} (DFT) is a tuple $T=(\Sigma,\Gamma,Q,q_0,\delta)$ where 
\begin{compactitem}
    \item $\Sigma$ and $\Gamma$ are the input and output alphabets,
    \item $Q$ is the finite set of states,
    \item $q_0 \in Q$ is the initial state,
    \item $\delta \colon Q \times (\Sigma \cup \eos) \to \Gamma^* \times Q$ is the transition function.
\end{compactitem}
The transition function $\delta$ extends to strings as follows: 
$\delta(q,\varepsilon) = (\varepsilon,q)$ and for $u \in \Sigma^*$ and $a \in \Sigma$, $\delta(q,ua) = (u'v',s)$ where $\delta(q,u) = (u',r)$ and $\delta(r,a) = (v',s)$ for some $r \in Q$.
Then for any $w \in \Sigma^*$, we say that $T$ transduces $w$ to $w'$ iff $\delta(q_0, w \eos) = (w', r)$ for some $r \in Q$.
We call a transduction \emph{sequential} if it is definable by a DFT.
\end{definition}

Next, we introduce several nested classes of transductions: rational, regular, and polyregular.
We first give examples of transductions in these classes and informal descriptions of the classes in terms of various transducers. 

In brief, a transduction is \emph{rational} if it is definable by a 
\emph{nondeterministic} transducer (the kind probably most familiar to NLP researchers, except we are assuming it is total functional). 
A transduction is \emph{regular} if it is definable by a 
\emph{two-way} transducer, which can be thought of as a rational transducer that can go back and forth on the input string, or a Turing machine with a read-only input tape and a write-only, one-way output tape. 

\begin{example}
\label{ex:mapreverse_mapduplicate} The following transductions are regular but not rational:
\begin{itemize}
\item \mapreverse: Reverse each substring between markers.
\begin{center}
$\text{\marker ab\marker cde\marker fg\marker } \mapsto \text{\marker ba\marker edc\marker gf\marker }$
\end{center}
\item \mapduplicate: Duplicate each substring between markers.
\begin{center}
$\text{\marker ab\marker cde\marker } \mapsto \text{\marker abab\marker cdecde\marker }$
\end{center}
\end{itemize}
\end{example}

A transduction is \emph{polyregular} if it is definable by a 
\emph{pebble} transducer, which is a 
two-way transducer augmented with a stack of up to $k$ \emph{pebbles} \citep{bojanczyk2022}.
It can push the current head position onto the stack and jump to the beginning of the string, and it can pop the top pebble from the stack and jump to that pebble's position. It can read the symbol at every pebble, and it can compare the positions of the pebbles. 
\begin{example} \label{ex:marked-squaring}
    The transduction $\markedsquare$ is polyregular but not regular. It makes $|w|$ many copies of $w$ separated by bars with successively longer prefixes marked, here by uppercasing:
    \begin{center}
        $\text{abaa} \mapsto \text{\marker  Abaa\marker  ABaa\marker  ABAa\marker  ABAA\marker }$
    \end{center}
\end{example}

Next we restrict to the \emph{aperiodic} subsets of these classes, and give formal definitions of these subclasses as composition closures of sets of \emph{prime transductions}. We will use these definitions for the rest of the paper.

\begin{definition}[aperiodicity]
Let $T$ be a deterministic finite automaton or transducer.
For any input \word{} $w$, there is a binary relation on states, $p \xrightarrow{w}_{T} q$, which holds iff $\delta(p,w)$ arrives at state $q$; if $T$ is a DFT, this means that $\delta(p, w) = (w', q)$ for some $w'$. 
Then $T$ is \emph{aperiodic} (or \emph{counter-free}) if  there is an $N \ge 0$ (depending on $T$) such that for all \words{} $w \in \Sigma^*$ and all $n \ge N$, the relations $\xrightarrow{w^n}_{T}$ and $\xrightarrow{w^{n+1}}_{T}$ are the same.
\end{definition}

Aperiodic deterministic finite automata (DFAs) are equivalent to star-free regular expressions and first-order logic with order \citep{schutzenberger1965,mcnaughton1971counter}.
They are also equivalent to masked hard-attention transformers
\citep{yang2024masked}. We take this equivalence as our starting point.

\begin{example}
The regular language $(\text{ab})^*$ is definable by an aperiodic DFA (with $N=2$):
\begin{center}
        \begin{tikzpicture}[x=2.5cm,y=2.5cm]
          \tikzset{state/.append style={minimum size=7mm,inner sep=0mm, outer sep=0mm}}
          \tikzset{every node/.append style={inner sep=0mm,outer sep=1mm}}
          \tikzset{every edge/.append style={->,>=stealth}}
          \node[initial,accepting,state](q1) {$q_1$};
          \node[state](q2) at (1,0) {$q_2$};
          \node[state](q3) at (0.5,-0.7) {$q_3$};
          \draw (q1) edge[bend left=15] node[above] {$\text{a}$} (q2) edge node[auto=right] {$\text{b}$} (q3);
          \draw (q2) edge[bend left=15] node[below] {$\text{b}$} (q1) edge node[auto=left] {$\text{a}$} (q3);
          \draw (q3) edge[loop below] node[below] {$\text{a}, \text{b}$} (q3);
        \end{tikzpicture}
\end{center}
But $(\text{aa})^*$ is not defined by any aperiodic DFA: the relations $\xrightarrow{\text{a}^n}_{T}$ and $\xrightarrow{\text{a}^{n+1}}_{T}$ are always different.
\end{example}

Each of the classes of transductions described above has an aperiodic subclass.
\begin{definition}
\label{def:aperiodic-transduction-classes}
\textit{\FOinit~sequential transductions} (which include string homomorphisms) are those defined by aperiodic DFTs.

\emph{\FOinit~rational transductions} are the composition closure of \FO~sequential transductions and \emph{right-to-left \FO{} sequential transductions}, that is, transductions that can be expressed as $w \mapsto f(w^\rev)^\rev$, where $f$ is \FO~sequential.%
\footnote{\Citet[fn.~xii]{nguyen-etal-2023} characterize \FO{} rational transductions using just one \FO{} sequential transduction and one right-to-left \FO{} sequential transduction, and \citet[Prop.~3]{filiot2016first} use a closely related characterization in terms of \emph{bimachines}. Here we use a composition of any number of transductions, which is equivalent because \FO{} rational transductions are closed under composition \citep[Thm.~10]{carton2021aperiodic}. }

\emph{\FOinit~regular transductions} are the composition closure of 
\FO~sequential transductions and the transductions
$\mapreverse$ and $\mapduplicate$ (\cref{ex:mapreverse_mapduplicate}).%
\footnote{This characterization is given by \citet[p.~15]{nguyenaperiodic}. It is also given by  \citet[Thm.~18]{BojanczykS20} for the more general setting of infinite alphabets constructed from \emph{atoms}; our definition here corresponds to the special case of finite alphabets (that is, where the set of atoms is empty).}

\emph{\FOinit~polyregular transductions} \citep[Def.~1.3]{bojanczyk2018} are the composition closure of 
\FO~regular transductions and the transduction $\markedsquare$
(\cref{ex:marked-squaring}).
\end{definition}

\subsection{Transformers}

We assume familiarity with transformers \citep{vaswani-etal-2017-attention} and describe a few concepts briefly. For more detailed definitions, please see the survey by \citet{strobl2023transformers}.

In standard attention, attention weights are computed from attention scores using the softmax function.
In \emph{average-hard attention} \citep{perez-etal-2021-turing,merrill-etal-2021-saturated-transformers}, each position $i$ attends to those positions $j$ that maximize the score $s_{i,j}$.
If there is more than one such position, attention is divided equally among them.
In \emph{unique-hard attention} \citep{hahn-2020-theoretical}, exactly one maximal element receives attention.
In leftmost hard attention, the leftmost maximum element is chosen, while in rightmost hard attention, the rightmost maximum element is chosen.

In this work, we use RASP~\citep{weiss-etal-2021-rasp} as a proxy for transformers.
Specifically, we use extensions of $\BRASP$, a version of RASP restricted to Boolean values \citep{yang2024masked}.
$\BRASP$ is equivalent to masked hard-attention transformer encoders with leftmost and rightmost hard attention, and with strict future and past masking.

\section{$\BRASP$ and unique hard attention transformers}

In this section, in order to facilitate our study of transformers and how they relate to classes of transductions,
we modify the definition of $\BRASP$ to compute transductions and use it to show that unique-hard attention transformers are equivalent to aperiodic rational transductions.

In \cref{sec:brasppos,sec:SRASP}, we consider two extensions:
$\BRASPpos$ adds position information, and $\SRASP$ also includes an operator for prefix sum. These have various correspondences both with more realistic transformers and with larger classes of transductions.

\subsection{Definition and examples}

We give an example first, followed by a more systematic definition.

\begin{example} \label{ex:increment}
The following $\BRASP$ program computes the transduction $\increment$, which takes as input a binary number (with its high-order bit on the left) and increments it, ignoring overflow.
\begin{raspcode}
    \begin{align*}
        \vecname{not}(i) &= \rechar{1} \iftext \inputx(i) = \rechar{0} \elsetext \rechar{0} \\
        \vecname{carry}(i) &= \attrdefault{j}{j>i}{\inputx(j) = \rechar{0}}{\false}{\true} \\
        \outputy(i) &= \vecname{not}(i) \iftext \vecname{carry}(i) \elsetext \inputx(i)
    \end{align*}
\end{raspcode}
\Cref{fig:steps-increment} shows a sample run.
The input string is stored in $\inputx(0), \ldots, \inputx(n-1)$. 
The vector $\vecname{not}$ is the bitwise negation of $\inputx$. The $\iftext$ expression is in Python-style syntax: if $\inputx(i) = \rechar{0}$, then $\vecname{not}(i) = \rechar{1}$; otherwise, $\vecname{not}(i) = \rechar{0}$. The vector $\vecname{carry}$ tests at each position $i$ whether there is a carry at that position, that is, whether at every position $j>i$ the symbol is a $\rechar{1}$. It can be read as: ``Find the rightmost ($\rightmost$) position $j$ such that $j>i$ and $\inputx(j) = \rechar{0}$. If there is such a position, return false ($\false$); if there is no such position, return true ($\true$).'' Finally, the vector~$\outputy$ is the output of the program.
\begin{table}[ht]
\begin{center}
\exampletablefontsize
\setlength{\tabcolsep}{5pt}
\begin{tabular}{ l|cccccccccc| } 
 \toprule
 \inputx         & 0 & 1 & 0 & 1 & 1 \\
 \vecname{not}   & 1 & 0 & 1 & 0 & 0 \\
 \vecname{carry} & 0 & 0 & 1 & 1 & 1 \\
 \outputy        & 0 & 1 & 1 & 0 & 0 \\
 \bottomrule
\end{tabular}
\end{center}
\caption{$\BRASP$ computation for $\increment$.}
\label{fig:steps-increment}
\end{table}
\end{example}

We give a definition of $\BRASP$ that is equivalent to that of \citet{yang2024masked}, extended to transductions.
For now, we consider $\BRASP$ programs for length-preserving transductions, and will later consider two schemes for defining non-length-preserving transductions as needed.

There are two types of values: Booleans from $\{\true,\false\}$, and symbols from a finite alphabet $\Delta$. These are stored in \emph{vectors}, which all share the same length~$n$, mirroring the transformer encoders they are intended to model.

A $\BRASP$ program receives an input string $w = a_0 \cdots a_{n-1}$ represented as a symbol vector $\inputx$, where $\inputx(i) = a_i$ for $i \in [n]$.

A $\BRASP$ program is a sequence of definitions of the form $P(i) = \rho$, where $P$ is a vector name, $i$ is a variable name, and $\rho$ is a right-hand side, to be defined below. The type of $P$ is the type of $\rho$.
No two definitions can have the same left-hand side.

The syntax of $\BRASP$ expressions, with Boolean ($\bool$) and symbolic ($\chr$) type, is:
\begin{align*}
e^\bool &\bnfto \true \mid \false \mid P^\bool(i) \mid e^\chr = e^\chr \\
& \bnfmid e^\bool \land e^\bool \mid e^\bool \lor e^\bool \mid \neg e^\bool \\
e^\chr &\bnfto \rechar{a} \mid \rechar{b} \mid \cdots \mid P^\chr(i) \\
&\bnfmid e^\chr \iftext e^\bool \elsetext e^\chr
\end{align*}
where $P$ is a vector name and $i$ is a variable name.
We write $\fv(e)$ for the variables occurring in $e$.
As mentioned above, conditional expressions use Python syntax: $e_1 \iftext e_2 \elsetext e_3$ means ``if $e_2$ evaluates to $\true$, then return $e_1$; otherwise, return $e_3$.'' The syntax of expressions could be extended to include arbitrary operations on Booleans or symbols.

Each definition has one of the following forms:
\begin{enumerate}
    \item \emph{Position-wise operations} 
    $P(i) = e$, where $e$ is an expression such that $\fv(e) \subseteq \{i\}$.
    
   \item \emph{Attention operations,} which have one of the two forms
   \begin{gather*}
   P(i) = \attldefault{j}{M(i,j)}{S(i,j)}{V(j)}{D(i)} \\  
   P(i) = \attrdefault{j}{M(i,j)}{S(i,j)}{V(j)}{D(i)}
   \end{gather*}
   where:  \begin{compactitem}
        \item The \emph{\direction} is either leftmost ($\leftmost$) or rightmost ($\rightmost$).
        \item $M(i,j)$ is a \emph{mask predicate}, one of
        \begin{compactenum}
            \item \emph{no masking}: $M(i,j) = \true$
            \item \emph{future masking}: $M(i,j) = (j < i)$ or $M(i,j) = (j \le i)$
            \item \emph{past masking}: $M(i,j) = (j > i)$ or $M(i,j) = (j \ge i)$.
        \end{compactenum}
        \item $S(i,j)$ is an \emph{attention predicate}, given by a Boolean expression with $\fv(S(i,j)) \subseteq \{i, j\}$. %
        \item $V(j)$ is a \emph{value function}, given by a Boolean or symbol expression with $\fv(V(j)) \subseteq \{j\}$. %
        \item $D(i)$ is a \emph{default function}, given by a Boolean or symbol expression with $\fv(D(i)) \subseteq \{i\}$. %
    \end{compactitem}
    The attention operation defines a new vector $P$, as follows. For $i \in [n]$ and \direction{} $\leftmost$, $j_i$ is the minimum $j \in [n]$ such that $M(i,j) = \true$ and $S(i,j) = \true$, if any, and $P(i)$ is set to the value $V(j_i)$.
    If there is no such $j$, then $P(i)$ is set to the value $D(i)$. If the \direction{} is $\rightmost$ then $j_i$ is the maximum $j \in [n]$ such that $M(i,j) = \true$ and $S(i,j) = \true$, if any, and $P(i)$ is set to the value $V(j_i)$. If there is no such $j$, then $P(i)$ is set to the value $D(i)$.
\end{enumerate}

The output of a $\BRASP$ program is given in a designated symbol vector $\outputy$, which has the same form as the input vector $\inputx$.

\begin{example}
\label{ex:rotateright}
The rational transduction $\rotateright$ rotates the input string to the right by one symbol, moving the last symbol to the first position. For example, 
\[\text{abc} \mapsto \text{cab}\]

The following $\BRASP$ program computes $\rotateright$:
\begin{raspcode}
    \begin{align*}
        \vecname{right}(i) &= \attldefault{j}{j>i}{\true}{\inputx(j)}{\rechar{\pad}}\\
        \vecname{last}(i) &= \attldefault{j}{\true}{\vecname{right}(j) = \rechar{\pad}}{ \inputx(j)}{\rechar{\pad}}\\
        \vecname{left}(i) &= \attrdefault{j}{j<i}{\true}{ \inputx(j)}{\rechar{\pad}}\\
        \outputy(i) &=  \vecname{left}(i) \iftext \vecname{left}(i) \neq \rechar{\pad} \elsetext \vecname{last}(i)
    \end{align*}
\end{raspcode}
An example run is in \cref{fig:steps-_for_rotate_right_one}.
The vector $\vecname{right}$, at each position $i$, records the symbol immediately to the right of $i$ (or $\rechar{\pad}$ if there is no symbol to the right).
We distinguish the position $j$ of the rightmost symbol in the input \word{} by testing whether $\vecname{right}(j) = \rechar{\pad}$, and propagate its input symbol to all positions in the vector $\vecname{last}$.
The vector $\vecname{left}$ records the symbol immediately to the left of each position (or $\rechar{\pad}$ if there is no symbol to the left).
To compute the output vector $\outputy$, the first position takes on the value of the rightmost symbol of the input \word{} and each other position takes on the value of its left neighbor, via a position-wise operation.
\end{example}

\begin{table}[ht]
\begin{center}
\exampletablefontsize
\setlength{\tabcolsep}{5pt}
\begin{tabular}{ l|cccccccccc| } 
 \toprule
 \inputx        & a & b & c & b & b & a & c \\
 \vecname{right}      & b  &  c  &  b  &  b  &  a  &  c  &  \pad \\
 \vecname{last}        & c  &  c  &  c  &  c  &  c  &  c  &  c  \\
 \vecname{left}       & \pad  &  a  &  b  &  c  &  b  &  b  &  a   \\
 \outputy        & c & a & b & c & b & b & a \\
 \bottomrule
\end{tabular}
\end{center}
\caption{$\BRASP$ computation for $\rotateright$.}
\label{fig:steps-_for_rotate_right_one}
\end{table}

\subsection{\Compressed{} outputs} 

So far, we have defined $\BRASP$ to encompass only length-preserving transductions.
But even some simple classes of transductions, like string homomorphisms, are not length-preserving.

To address this, we allow the program to output a vector containing strings up to some length $k$ instead of a vector of symbols.
For any finite alphabet $A$, let $A^{\le k}$ denote the set of all strings over $A$ of length at most $k$ (including the empty string $\varepsilon$).

The input vector is still a vector of input symbols: $a_0 a_1 \cdots a_{n-1}$, where $a_i \in \Sigma$ for $i \in [n]$.
However, the output vector is a vector of symbols over the alphabet $\Gamma^{\le k}$ for some $k$.
The output vector is a
\emph{$k$-\compressed{} representation} of a \word{} $u$ if the concatenation of the strings at positions $0, \ldots, n-1$ is $u$.
There may be many different $k$-\compressed{} representations of the same string.
For an input string of length $n$, the output string has length at most $kn$.
\Compressed{} outputs make it possible to compute any string homomorphism, as in the following example.
\begin{example}
\label{ex:homomorphism}
Apply the homomorphism $\text{a} \mapsto \text{aa}$, $\text{b} \mapsto \text{ccb}$ to an input \word{} over the alphabet $\{\text{a},\text{b}\}$.
\begin{raspcode}
    \begin{align*}
        \outputy(i) &= \rechar{aa} \iftext \inputx(i) = \rechar{a} \\
                    &\phantom{={}} \elsetext \rechar{ccb} %
    \end{align*}
\end{raspcode}
\end{example}

\subsection{$\BRASP$ defines exactly the \FO{} rational transductions}

\Cref{ex:increment,ex:rotateright} show that $\BRASP$ can compute some \FO{} rational transductions that are not sequential.
The following theorem shows that $\BRASP$ can compute only \FO{} rational transductions.

\begin{theorem}
    \label{theorem:BRASP-is-FOTop}
    Any $\BRASP$ program with \compressed{} outputs defines an \FO{} rational transduction.
\end{theorem}

\begin{proof} \newcommand{\val}{x}
Let $\progname{P}$ be a $\BRASP$ program.
By Lemma~12 of \citet{yang2024masked}, $\progname{P}$ can be rewritten so that every score predicate $S(i,j)$ depends only on~$j$.
Denote the sequence of vectors of $\progname{P}$ as $P_1, \ldots, P_m$, and treat the input vector $\inputx$ as $P_0$.
We prove by induction that the first $k$ operations of $\progname{P}$ can be converted to a composition of left-to-right and right-to-left \FO~sequential transductions.
The output of the composition is the sequence of $(k+1)$-tuples $(\val_{0,0}, \ldots, \val_{0,k}), \ldots, (\val_{n-1,0}, \ldots, \val_{n-1,k})$, where for $i \in [n]$ and $j \in [k+1]$, we have $\val_{i,j} = P_j(i)$.

If $k=1$, we just construct the identity transducer. If $k>1$, assume that the first $k-1$ operations have been converted to a composition of transductions.
If $P_k$ is a position-wise operation, it can be computed by a one-state DFT that appends the value of $\val_{i,k} = P_k(i)$ onto the end of the input $k$-tuple.  The interesting cases are the attention operations.

Case $P_k(i) = \attrdefault{j}{j<i}{P_s(j)}{P_v(j)}{P_d(i)}$, where $s, v, d < k$: Let $T$ be the set of values in the type of $P_k$. Then we construct the following (left-to-right) DFT. Starting from the first position, it appends $P_d(i)$ onto the end of the input $k$-tuple. Every time it reaches a position $j$ where $P_s(j)$ is true, it switches, starting from position $j+1$, to appending $P_v(j)$.
In the following, $\vec{\val}$ is the input $k$-tuple, $x_r$ is the component of $\vec{\val}$ with index $r$, and $(\vec{\val},x)$ is the $(k+1)$-tuple obtained by appending the element $x$ to the end of $\vec{\val}$.
\begin{align*}
Q &= \{ q_{\textnormal{def}} \} \cup \{ q_\val \mid \val \in T \} \\
\delta(q_{\textnormal{def}}, \vec{\val}) &= \begin{cases}
((\vec{\val}, \val_d), q_{\val_v}) & \val_s = \true \\
((\vec{\val}, \val_d), q_{\textnormal{def}}) & \val_s = \false
\end{cases} \\
\delta(q_x, \vec{\val}) &= \begin{cases}
\mathrlap{((\vec{\val}, \val), q_{\val_v})}
\hphantom{((\vec{\val}, \val_d), q_{\textnormal{def}})}
& \val \in T, \val_s = \true \\
((\vec{\val}, \val), q_{\val}) & \val \in T, \val_s = \false.
\end{cases}
\end{align*}
To see that this is counter-free: Let $u$ be any string. If $u$ contains a tuple $\vec{\val}$ such that $x_s = \true$, let $\vec{\val}$ be the rightmost such tuple. Then $q \xrightarrow{u} q_{\val_v}$ for all $q$, so $(\xrightarrow{u^i}) = (\xrightarrow{u^{i+1}})$ for all $i \ge 1$. If $u$ does not contain such a tuple, then $q \xrightarrow{u} q$ for all $q$, so $(\xrightarrow{u^i}) = (\xrightarrow{u^{i+1}})$ for all $i \ge 0$. 

Case $P_k(i) = \attldefault{j}{j<i}{P_s(j)}{P_v(j)}{P_d(i)}$, where $s, v, d < k$: Let $Q$ and $T$ be as above. Then we construct the following DFT. Starting from the first position, it appends $P_d(i)$ onto the end of the input $k$-tuple. The first time it reaches a position $j$ where $P_s(j)$ is true, it switches to appending $P_v(j)$, from position $j+1$ to the end.
\begin{align*}
\delta(q_{\textnormal{def}}, \vec{\val}) &= \begin{cases}
((\vec{\val}, \val_d), q_{\val_v}) & \val_s = \true \\
((\vec{\val}, \val_d), q_{\textnormal{def}}) & \val_s = \false
\end{cases} \\
\delta(q_x, \vec{\val}) &= \mathrlap{((\vec{\val}, \val), q_\val)} 
\hphantom{\begin{cases} ((\vec{\val}, \val_d), q_{\textnormal{def}}) & \end{cases}}
x \in T.
\end{align*}
To see that this is counter-free: Same as the previous case, except $\vec{\val}$ is the \emph{leftmost} tuple in $u$ such that $v_s = \true$.  %

The cases 
\begin{align*}
P_k(i) &= \attldefault{j}{j>i}{P_s(j)}{P_v(j)}{P_d(i)} \\
P_k(i) &= \attrdefault{j}{j>i}{P_s(j)}{P_v(j)}{P_d(i)}
\end{align*}
are the same, but using right-to-left transducers.

Case $P_k(i) = \attldefault{j}{\true}{P_s(j)}{P_v(j)}{P_d(i)}$, where $s, v, d < k$: This operation could be replaced by the following sequence of three operations, which are covered in the preceding cases.
\begin{align*}
R_k(i) &= \attldefault{j}{j>i}{P_s(j)}{P_v(j)}{P_d(i)}\\
C_k(i) &= P_v(i) \iftext P_s(i) \elsetext R_k(i) \\
P_k(i) &= \attldefault{j}{j<i}{P_s(j)}{P_v(j)}{C_k(i)}.
\end{align*}
Here, $R_k(i)$ is the value from the leftmost $j > i$ with $P_s(j) = \true$ (if any), else $P_d(i)$; then
$C_k(i)$ is the value from the leftmost $j \ge i$ with $P_s(j) = \true$ (if any), else $P_d(i)$; finally,
$P_k(i)$ is the value from the leftmost $j$ overall with $P_s(j) = \true$ (if any), else $P_d(i)$.

Case $P_k(i) = \attrdefault{j}{\true}{P_s(j)}{P_v(j)}{P_d(i)}$ is the mirror image of the previous case.
\end{proof}

For the converse, we need the following lemma.
\begin{lemma} \label{thm:flipflop_to_brasp}
If\/ $\progname{P}$ is a \BRASP{} program with \compressed{} outputs and $f$ is an \FO{} sequential transduction, there is a \BRASP{} program with \compressed{} outputs that computes $f \circ \progname{P}$.
\end{lemma}
\begin{proof}
We can adapt the proof of Lemma 19 of \citet{yang2024masked}. 
By the Krohn-Rhodes decomposition theorem for \FO{} sequential transductions \citep[Thm.~4.8]{pradicnguyen2020}, $f$ is equivalent to the sequential composition of finitely many two-state aperiodic DFTs.
Hence, without loss of generality, we can assume that $f$ is defined by a two-state aperiodic DFT $T$.
This machine $T$ is an \emph{identity--reset} transducer, which means that for any symbol~$\sigma \in \Sigma$, the state transformation $\xrightarrow{\sigma}_T$ either is the identity (maps both states to themselves) or \emph{resets} to one of the states $q$ (maps both states to $q$).
For each state $q$ of~$T$, let $R_q$ be the set of symbols that reset to $q$. Let $R = \bigcup_q R_q$ and $I = \Sigma \setminus R$.
Let~$q_1$ be the start state and $q_2$ the other state.
We write $T(q, w) = w'$ if $\delta(q, w) = (w', q')$ for some~$q'$.

Modify $\progname{P}$ so that its output vector is a fresh vector $\vecname{z}$ instead of $\outputy$.
Then $f \circ \progname{P}$ is defined by appending the following operations to $\progname{P}$:
\begin{raspcode}
\begin{align*}
\vecname{state}_q(i) &= \rightmost_j \left[j<i, \bigvee_{\substack{uav \in \Gamma^{\le k} \\ a \in R, v \in I^*}} \vecname{z}(j) = uav\right] \\
&\phantom{={}} \left(\bigvee_{\substack{uav \in \Gamma^{\le k} \\ a \in R_q, v \in I^*}} \vecname{z}(j) = uav\right) : q = q_1 \\
\vecname{sym}(i) &= \attldefault{j}{j>i}{\true}{\vecname{z}(i)}{\vecname{z}(i) \eos} \\
\outputy(i) &= T(q_1, \vecname{sym}(i)) \iftext \vecname{state}_{q_1}(i) \\
&\phantom{={}} \elsetext T(q_2, \vecname{sym}(i)) %
\end{align*}
\end{raspcode}
Vector $\vecname{state}_q(i)$ tests whether $T$ is in state $q$ just before reading (\compressed{}) symbol $w_i$. It does so by searching for the rightmost symbol $a$ that resets to any state. If $a$ exists and resets in particular to $q$, then $T$ must still be in state $q$; otherwise, it is not. But if $a$ does not exist, then $T$ must still be in the start state $q_1$.
Vector $\vecname{sym}(i)$ simply appends $\eos$ to the last position.
Finally, $\outputy$ maps $\vecname{sym}(i)$ to $T(q, \vecname{sym}(i))$ (where $q$ is the state just before reading $w_i$).
\end{proof}

\begin{table}[ht]
\begin{center}
\exampletablefontsize
\setlength{\tabcolsep}{5pt}
\begin{tabular}{@{}l|c|ccccccccccc@{}} 
  \toprule
  \inputx     &  input     & | & a & b & | & c & d & e & | \\
  \pos   & position        & $0$ & $1$ & $2$ & $3$ & $4$ & $5$ & $6$ & $7$ \\
  \vecname{prev}   & previous \rechar{|} & $0$ & $0$ & $0$ & $0$ & $3$ & $3$ & $3$ & $3$ \\
  \vecname{next}   & next \rechar{|} & $3$ & $3$ & $3$ & $7$ & $7$ & $7$ & $7$ & $0$ \\
  \vecname{src}   & source & $3$ & $2$ & $1$ & $4$ & $6$ & $5$ & $4$ & $0$ \\
  $\vecname{y1}$   & $\inputx(\vecname{src}(i))$  & | & b & a & c & e & d & c & |\\
  \outputy    &  output & | & b & a & | & e & d & c & | \\
  \bottomrule
 \end{tabular}
 \end{center}
 \caption{Example $\BRASPpos$ computation for $\mapreverse$. Details in Ex.~\ref{ex:mapreverse_rasp}.}
 \label{tab:map-reverse}
\end{table}

\begin{table}[ht]
\begin{center}
\exampletablefontsize
\setlength{\tabcolsep}{5pt}
\begin{tabular}{@{}l|c|cccccccc@{}}
    \toprule
    \inputx &  input & | & a & b & | & c & d & e & | \\
    \vecname{pos} & position  & 0 & 1 & 2 & 3 & 4 & 5 & 6 & 7 \\
    \vecname{prev} & previous~\rechar{|} & 0 & 0 & 0 & 0 & 3 & 3 & 3 & 3 \\
    \vecname{next} & next \rechar{|} & 3 & 3 & 3 & 7 & 7 & 7 & 7 & 0 \\
    \vecname{nowrap} &  & 0 & 1 & 3 & 5 & 4 & 6 & 7 & 7 \\
    \vecname{wrap} &  & 0 & 0 & 1 & 0 & 1 & 3 & 5 & 7 \\
    \vecname{src1} & left symbol & 0 & 1 & 1 & 5 & 4 & 6 & 5 & 7 \\
    \vecname{src2} & right symbol & 1 & 2 & 2 & 6 & 5 & 4 & 6 & 7 \\
    \outputy & output & | & \clap{ab} & \clap{ab} & | & \clap{cd} & \clap{ec} & \clap{de} & | \\
    \bottomrule
 \end{tabular}
 \end{center}
 \caption{Example $\BRASPpos$ computation for $\mapduplicate$. Details in Ex.~\ref{ex:mapduplicate}.}
 \label{tab:map-duplicate-alternate}
\end{table}

\begin{theorem} \label{thm:arational_to_brasp}
For any \FO{} rational transduction $f \colon \Sigma^* \to \Gamma^*$, there is a \BRASP{} program $\progname{P}$ with \compressed{} outputs that computes $f$.
\end{theorem}

\begin{proof}
The transduction $f$ can be written as $f_R \circ f_L$, where $f_L$ is an \FO{} sequential transduction and $f_R$ is a right-to-left \FO{} sequential transduction (\cref{def:aperiodic-transduction-classes}). 
The identity transduction can clearly be computed by a \BRASP{} program, and by \cref{thm:flipflop_to_brasp} there is a \BRASP{} program computing $f_L$. Finally, \cref{thm:flipflop_to_brasp} can be easily modified to  apply also to $f_R$, using the mirror images of $\vecname{state}_q$ and $\vecname{sym}$ above.
\end{proof}

\subsection{Unique-hard attention transformers compute exactly the aperiodic rational transductions}

\Citet{yang2024masked} show that a $\BRASP$ program can be simulated by a unique-hard attention transformer with no position information besides masking, and vice versa. With \cref{theorem:BRASP-is-FOTop,thm:arational_to_brasp}, this implies that masked unique-hard attention transformers with \compressed{} outputs can compute exactly the \FO{} rational transductions. 

\section{$\BRASP$ with positions}
\label{sec:brasppos}

\subsection{Definition}

We extend $\BRASP$ to $\BRASPpos$,
which adds a type $\nat$ for integers in $[n]$, and vectors containing integers.

We extend the syntax of expressions as follows:
\begin{align*}
e^\nat &\bnfto 0 \mid 1 \mid e^\nat + e^\nat \mid e^\nat - e^\nat \\
e^\bool &\bnfto \cdots \\
&\bnfmid c^\bool \qquad \text{where $|\fv(c^\bool)| \le 1$} \\
c^\bool &\bnfto e^\nat = e^\nat \mid e^\nat < e^\nat \mid e^\nat \le e^\nat \\
&\bnfmid e^\nat > e^\nat \mid e^\nat \ge e^\nat
\end{align*}
where $\cdots$ means all of the productions from the syntax of $\BRASP$.
Then vector definitions are extended as follows.
\begin{compactenum}
    \item There is a pre-defined integer vector $\pos(i)$, whose value is simply $i$ at every position $i \in [n]$.
    \item There are position-wise operations $P(i) = e^\nat$, where $\fv(e^\nat) \subseteq \{i\}$. Addition and subtraction have their usual meaning, but values less than $0$ are replaced by $0$ and values greater than $n-1$ are replaced by $n-1$. (Since this is not associative, we fully parenthesize arithmetic expressions.)
    \item There are position-wise operations $P(i) = c^\bool$, where $\fv(c^\bool) \subseteq \{i\}$. The operators $<$, $>$, $=$, $\neq$, $\le$, and $\ge$ have their usual meaning.
    \item In $\BRASP$, $S(i,j)$ was a Boolean expression ($e^\bool$); in $\BRASPpos$, it can be either a Boolean expression ($e^\bool$) or, as a special case, $V_1(i) = V_2(j)$, where $V_1$ and $V_2$ are previously defined integer vectors. We emphasize that only tests for equality are allowed (not, for example, $V_1(i) < V_2(j)$). This restriction is used in the transformer simulation in \cref{sec:S-RASP-with-augmented-average-hard-attention}.
\end{compactenum}

\subsection{Examples}

Informally, we omit a default value from a leftmost or rightmost operation if the operation is such that the default value will never be taken.

\begin{example}[$\mapreverse$]
\label{ex:mapreverse_rasp}
    Reverse each substring between markers.
\begin{center}
$\text{\marker ab\marker cde\marker fg\marker } \mapsto \text{\marker ba\marker edc\marker gf\marker }$
\end{center}
\begin{raspcode}
\begin{align*}
        \vecname{prev}(i) &= \attrdefault{j}{j < i}{\inputx(j) = \rechar{|}}{\pos(j)}{0} \\
        \vecname{next}(i) &= \attldefault{j}{j>i}{\inputx(j)=\rechar{|}}{\pos(j)}{0} \\
        \vecname{src}(i) &= \vecname{prev}(i)+\vecname{next}(i)-\pos(i) \\
        \vecname{y1}(i) &= \attl{j}{\true}{\vecname{src}(i) = \pos(j)}{\inputx(j)} \\
        \outputy(i) &= \rechar{|} \iftext \inputx(i) = \rechar{|} \elsetext \vecname{y1}(i)
\end{align*}
\end{raspcode}
    An example run is in \cref{tab:map-reverse}.
\end{example}

Above, the vector $\vecname{y1}$ just retrieves, for each $i$, the input symbol at position $\vecname{src}(i)$. This idiom is so common that we will write it using the syntactic sugar:
    \begin{raspcode}\[\vecname{y1}(i) = \inputx(\vecname{src}(i)).\]\end{raspcode}

\begin{example}[$\mapduplicate$]
\label{ex:mapduplicate}
Duplicate each substring between markers.
\begin{center}
$\text{\marker ab\marker cde\marker } \mapsto \text{\marker abab\marker cdecde\marker }$
\end{center}
\begin{raspcode}
\begin{align*}
\vecname{prev}(i) &= \attrdefault{j}{j<i}{\inputx(j)=\rechar|}{\pos(j)}{0} \\
\vecname{next}(i) &= \attldefault{j}{j>i}{\inputx(j)=\rechar|}{\pos(j)}{0} \\
\vecname{nowrap}(i) &= \pos(i) + (\pos(i) - \vecname{prev}(i) - 1) \\
\vecname{wrap}(i) &= \pos(i) - (\vecname{next}(i) - \pos(i)) \\
\vecname{src1}(i) &= \vecname{nowrap}(i)
                   \iftext \vecname{nowrap}(i) < \vecname{next}(i) \\ &\phantom{={}} \elsetext \vecname{wrap}(i) \\
\vecname{src2}(i) &= \vecname{nowrap}(i)+1 %
                   \iftext \vecname{nowrap}(i)+1 < \vecname{next}(i) \\
                   &\phantom{={}} \elsetext \vecname{wrap}(i)+1 \\
\outputy(i) &= \rechar| \iftext \inputx(i) = \rechar| \\
&\phantom{={}} \elsetext \inputx(\vecname{src1}(i)) \conc \inputx(\vecname{src2}(i))
\end{align*}
\end{raspcode}
Here $\conc$ denotes string concatenation over $\Gamma^{\le k}$.
An example run is in \cref{tab:map-duplicate-alternate}. Note that $\vecname{nowrap}(6) = 7$, not $8$, because addition and subtraction are clipped to lie in $[0, n-1]$.
\end{example}

\begin{example}[$\copyfirsthalf$]
\label{ex:copyfirsthalf}
    Copy just the first half of the input \word{}, rounding down.
    \begin{center}
    $\text{abcaabcbb} \mapsto \text{abca}$
    \end{center}
    \begin{raspcode}
    \begin{align*}
        \vecname{last}(i) &= \attr{j}{\true}{\true}{\pos(j)} \\
        \vecname{sum}(i) &= \pos(i)+\pos(i) \\
        \outputy(i) &= \inputx(i) \iftext \vecname{sum}(i) < \vecname{last}(i) \elsetext \charepsilon
    \end{align*}
    \end{raspcode}

    An example run is in \cref{tab:copy-first-half-rounding-down}.
    \begin{table}[ht]
\begin{center}
\exampletablefontsize
\setlength{\tabcolsep}{5pt}
\begin{tabular}{@{}l|c|ccccccccc@{}}
  \toprule
  $\inputx$        & input     & a & b & c & a & a & b & c & b & b \\
  $\pos$   & $i$      & $0$ & $1$ & $2$ & $3$ & $4$ & $5$ & $6$ & $7$ & $8$\\
  \vecname{last} &   $n-1$              & $8$ & $8$ & $8$ & $8$ & $8$ & $8$ & $8$ & $8$ & $8$ \\
  \vecname{sum}  & $\min(2i,n-1)$     & $0$ & $2$ & $4$ & $6$ & $8$ & $8$ & $8$ & $8$ & $8$ \\
  \outputy     & output              & a & b & c & a & $\varepsilon$ & $\varepsilon$ & $\varepsilon$ & $\varepsilon$ & $\varepsilon$\\
  \bottomrule
 \end{tabular}
 \end{center}
 \caption{Example $\BRASPpos$ computation for $\copyfirsthalf$. Details in Ex.~\ref{ex:copyfirsthalf}.}
 \label{tab:copy-first-half-rounding-down}
\end{table}

\end{example}
\begin{proposition} \label{thm:copyfirsthalf}
    The transduction $\copyfirsthalf$ is neither regular nor polyregular.
\end{proposition}
\begin{proof}
    Both regular and polyregular transductions preserve regular languages under inverse \citep[Thm 1.7]{bojanczyk2018}. 
    The inverse of the regular language $a^*$ under $\copyfirsthalf$ is the set of words of the form $a^nw$ where $|w| \le n$, which is not regular if $|\Sigma| > 1$, so the transduction $\copyfirsthalf$ is neither regular nor polyregular.\footnote{Thanks to an anonymous reviewer for suggesting this argument to us.}
\end{proof}

\begin{example}[\residuesmodm]
Let $m$ be a positive integer and define the transduction \residuesmodm{} to map any input $a_0a_1 \cdots a_{n-1}$ to the sequence $b_0 b_1 \cdots b_{n-1}$ where $b_i = i \bmod m$. This transduction is rational but not aperiodic.
\end{example}
\begin{proposition}
    \label{prop:BRASPpos-computes-residues}
    For any $m$, \BRASPpos{} can compute the transduction \residuesmodm{}.
\end{proposition}

\begin{proof}
For concreteness, we give a program for the case $m = 3$, which is easily generalized.
The second line deals with clipping to $n-1$.
\begin{raspcode}
\begin{align*}
    \vecname{sum3}(i) &= \pos(i) + \pos(i) + \pos(i) \\
    \vecname{sum3c}(i) &= \vecname{sum3}(i) \iftext \vecname{sum3}(i) = \vecname{sum3}(i-1) + 3 \\
    &\phantom{={}} \elsetext 0 \\
    \vecname{mult3}(i) &= \attldefault{j}{\true}{\pos(i) = \vecname{sum3c}(j)}{\true}{\false} \\
    \outputy(i) &= 0 \iftext \vecname{mult3}(i) \elsetext 1 \iftext \vecname{mult3}(i-1) \\
    &\phantom{={}} \elsetext 2 \tag*{\qedhere}
\end{align*}
\end{raspcode}
\end{proof}

\subsection{Expressivity}

\begin{theorem} \label{thm:aregular_to_brasppos}
\BRASPpos{} programs with \compressed{} outputs can compute all \FO{} regular transductions.
\end{theorem}

\begin{proof}
If $f$ is an \FO{} regular transduction, then by \cref{def:aperiodic-transduction-classes}, it can be decomposed into a composition of transductions, each of which is (a) \FO{} sequential, (b) $\mapreverse$, or (c) $\mapduplicate$. We convert $f$ to a \BRASPpos{} program by induction on the number of functions in the composition. Case (a) is the same as the proof of \cref{thm:flipflop_to_brasp}, \emph{mutatis mutandis}.
The following two lemmas handle the other two cases.
\end{proof}

\begin{lemma} \label{thm:compose_brasp_mapreverse}
If $\progname{P}$ is a \BRASPpos{} program with \compressed{} outputs, then there is a \BRASPpos{} program with \compressed{} outputs that computes $\mathord{\mapreverse} \circ \progname{P}$.
\end{lemma}

\begin{proof}
We'd like to compose $\progname{P}$ with the program of \cref{ex:mapreverse_rasp}, but since $\progname{P}$ uses \compressed{} outputs, we must adapt \cref{ex:mapreverse_rasp} to use \compressed{} inputs.
Define the functions $\transname{head}$, $\transname{body}$ and $\transname{tail}$ as follows. If $w$ does not contain the separator ($|$), then $\mathit{head}(w) = \mathit{tail}(w) = w$ and $\mathit{body}(w) = \varepsilon$.  Otherwise, factor $w$ as $xyz$, where $x$ is the prefix of $w$ before the first separator and $z$ is the suffix of $w$ after the last separator, and $\mathit{head}(w) = x$, $\mathit{body}(w) = y$, and $\mathit{tail}(w) = z$.
Position-wise operations allow the application of these functions, as well as $\mapreverse$ itself and the test of whether a string contains a symbol ($w \mapsto (a \in w)$), to bounded-length strings.
Modify $\progname{P}$ so that its output vector is a fresh vector $\vecname{z}$ instead of $\outputy$. Then append the following operations to $\progname{P}$:

\begin{raspcode}
\begin{align*}
\vecname{prev}(i) &= \attrdefault{j}{j<i}{\rechar{|} \in \vecname{z}(j)}{\pos(j)}{n-1} \\
\vecname{next}(i) &= \attldefault{j}{j>i}{\rechar{|} \in \vecname{z}(j)}{\pos(j)}{0} \\
\vecname{head}(i) &= \transname{head}(\vecname{z}(i)) \\
\vecname{body}(i) &= \transname{body}(\vecname{z}(i)) \\
\vecname{tail}(i) &= \transname{tail}(\vecname{z}(i)) \\
\vecname{nosep}(i) &= \vecname{z}(\vecname{next}(i) - (\pos(i)-\vecname{prev}(i)))^\rev \\
\vecname{ptail}(i) &= \vecname{tail}(\vecname{prev}(i))^\rev \\
\vecname{rbody}(i) &= \mapreverse(\vecname{body}(i)) \\
\vecname{nhead}(i) &= \vecname{head}(\vecname{next}(i))^\rev \\
\vecname{sep}(i) &= \vecname{ptail}(i) \conc \vecname{rbody}(i) \conc \vecname{nhead}(i) \\
\outputy(i) &= \vecname{sep}(i) \iftext \rechar{|} \in \vecname{z}(i) \elsetext \vecname{nosep}(i) \qedhere
\end{align*}
\end{raspcode}

To see how this works, consider a \compressed{} symbol $\vecname{z}(i)$. 
If it contains at least one separator, it is parsed into $\mathit{head}$, $\mathit{body}$, and $\mathit{tail}$ as $xyz$.  The correct output for position $i$ is computed in $\vecname{sep}(i)$, and consists of replacing $x$ by the reverse of the tail of the closest left neighbor that has a separator, replacing $y$ with $\mapreverse(y)$, and replacing $z$ by the reverse of the head of the closest right neighbor that has a separator.
If $\vecname{z}(i)$ contains no separator, then it appears in a maximal subsequence $w_0, w_1, \ldots, w_{k-1}$ with no separator, say as $w_{\ell}$, and the correct output for position $i$ is computed in $\vecname{nosep}$, and consists of the reverse of $w_{k-1-\ell}$.
\end{proof}

\begin{lemma}
If $\progname{P}$ is a \BRASPpos{} program with \compressed{} outputs, then there is a \BRASPpos{} program with \compressed{} outputs that computes $\mathord{\mapduplicate} \circ \progname{P}$.
\end{lemma}

\begin{proof}
As in the proof of \cref{thm:compose_brasp_mapreverse}, we want to compose $\progname{P}$ with the program in \cref{ex:mapduplicate}, so we adapt \cref{ex:mapduplicate} to use \compressed{} inputs.
Modify $\progname{P}$ so that its output vector is a fresh vector $\vecname{z}$ instead of $\outputy$. First append the following operations to $\progname{P}$ (where $\vecname{prev}$, $\vecname{next}$, $\vecname{head}$, $\vecname{body}$, and $\vecname{tail}$ are as in the proof of \cref{thm:compose_brasp_mapreverse}):

\begin{raspcode}
\begin{align*}
    \vecname{ptail}(i) &= (\vecname{tail}(i-1) \iftext i>0 \elsetext \charepsilon) \conc {} \vecname{head}(i) \\
    \vecname{nhead}(i) &= \vecname{tail}(i) \conc (\vecname{head}(i+1) \iftext i<n-1 \elsetext \charepsilon) \\
    \vecname{dbody}(i) &= \mapduplicate(\vecname{body}(i)) \\
    \vecname{sep}(i) &= \vecname{ptail}(i) \conc \vecname{dbody}(i) \conc \vecname{nhead}(i)
\end{align*}
\end{raspcode}

This computes in the vector $\vecname{sep}$ the correct outputs for those positions $i$ that have a separator in the input symbol $\vecname{z}(i)$.
The symbol is parsed into $\mathit{head}$, $\mathit{body}$ and $\mathit{tail}$ as $xyz$, and the correct output is the concatenation of the $\mathit{tail}$ of the preceding symbol,
the strings $x$, $\mapduplicate(y)$, $z$, and the $\mathit{head}$ of the the following symbol.
Note that $\mapduplicate$ is applied only to strings of bounded length.

The outputs for positions $i$ where $\vecname{z}(i)$ does not contain a separator are computed in the vector $\vecname{nosep}$ and combined with the values in $\vecname{sep}$ to produce the final output by the following operations.
\begin{raspcode}
\begin{align*}
    \vecname{nowrap}(i) &= \pos(i) + (\pos(i)-\vecname{prev}(i)) \\
    \vecname{wrap}(i) &= \pos(i) - (\vecname{next}(i)-\pos(i)) + 1) \\
    \vecname{half}(i) &= (\pos(i) - \vecname{prev}(i)) \le (\vecname{next}(i) - \pos(i))\\
    \vecname{src1}(i) &= \vecname{nowrap}(i) \iftext \vecname{half}(i) \elsetext \vecname{wrap}(i) \\
    \vecname{src2}(i) &= \vecname{src1}(i)+1 \iftext \vecname{src1}(i) < \vecname{next}(i)\\
                      &\phantom{={}} \elsetext \vecname{prev}(i) \\
    \vecname{sym1}(i) &= \vecname{tail}(\vecname{src1}(i)) \iftext \vecname{src1}(i) < \vecname{next}(i) \\
                      &\phantom{={}} \elsetext \vecname{head}(\vecname{src1}(i)) \\
    \vecname{sym2}(i) &= \vecname{tail}(\vecname{src2}(i)) \iftext \vecname{src2}(i) < \vecname{next}(i) \\
                      &\phantom{={}} \elsetext \vecname{head}(\vecname{src2}(i)) \\
    \vecname{nosep}(i) &= \vecname{sym1}(i) \conc \vecname{sym2}(i) \\
    \outputy(i) &= \vecname{sep}(i) \iftext \rechar{|} \in \vecname{z}(i) \elsetext \vecname{nosep}(i)
\end{align*}
\end{raspcode}
If the input symbol does not contain a separator, it is the concatenation of the symbols from $\vecname{sym1}$ and $\vecname{sym2}$, whose positions are calculated using the vectors $\vecname{nowrap}$ and $\vecname{wrap}$ in a manner similar to \cref{ex:mapduplicate}, but also including the $\mathit{tail}$ of the closest symbol on the left with a separator, and the $\mathit{head}$ of the closest symbol on the right with a separator.
\end{proof}

On the other hand, every operation in $\BRASPpos$ is computable by a family of $\AC^0$ circuits, that is, a family of Boolean circuits of constant depth and polynomial size \cite{hao-etal-2022-circuits}, which implies that any transduction computable in $\BRASPpos$ is computable in $\AC^0$.

\section{$\SRASP$ and average hard attention transformers}
\label{sec:SRASP}

\subsection{Definition}
We further extend $\BRASPpos$ to \emph{RASP with prefix sum} (or $\SRASP$) by adding a prefix sum operation.
\begin{definition}[Prefix sum]
\label{def:prefix-sum} A prefix sum operation has the form
    \[P(i) = \attsum{j}{j \le i}{V(j)}\]
    where $V(j)$ is an integer expression with $\fv(V(j)) \subseteq \{j\}$.
    It defines an integer vector $P(i)$ containing the sum of the values $V(j)$ for those positions $j$ such that $j \le i$. As with arithmetic operations, if the value of the prefix sum at a position is greater than $n-1$, it is replaced with $n-1$.
\end{definition}

\subsection{Padded inputs}

We defined non-length-preserving transductions for $\BRASP$ and $\BRASPpos$ by employing the convention of \compressed{} outputs.
However, for $\SRASP$, we introduce a simpler scheme: using only symbol, not string, vectors, while assuming that the input string is followed by padding symbols ${\pad}$, enough to accommodate the output string.

The input vector is $a_0 a_1 \cdots a_{\ell-1} \pad^{n-\ell}$, where $\ell<n$ and $a_i \in \Sigma$ for $i \in [\ell]$. The output vector, similarly, is $b_0 b_1 \cdots b_{k-1} \pad^{n-k}$, where $k<n$ and $b_i \in \Gamma$ for $i \in [k]$.

With this input/output convention, padding symbols may be necessary to create enough positions to hold the output \word{}.
But in order to be able to prove closure under composition for transductions computable in $\BRASP$ and its extensions, we allow additional padding symbols
to be required.
In particular, the program $\progname{P}$ \emph{computes the transduction} $f$ iff there exists a nondecreasing 
function
$q$, called the \emph{minimum vector length}, such that for every input \word{} $w \in \Sigma^{\ell}$, we have $q(\ell) \ge k = |f(w)|$, and if $\progname{P}$ is run on $ w\cdot\pad^{n-\ell}$, where $n > q(\ell)$, then the output is $f(w) \cdot \pad^{n-k}$.
In all of the examples in this section, except $\markedsquare$, $q$ is linear.

We could have used padded inputs with $\BRASP$ programs, but it can be shown that programs would only be able to map input strings of length $n$ to output strings of length at most $n+k$, for some constant $k$. \Compressed{} outputs give $\BRASP$ the ability to define transductions with longer outputs, like string homomorphisms.
However, the situation is exactly opposite with $\SRASP$. \Compressed{} outputs do not add any power to $\SRASP$, because ``unpacking'' a \compressed{} output into a vector of output symbols can be computed within $\SRASP$ (\cref{lemma:SRASP-computes-homomorphisms}).
Moreover, \compressed{} outputs only allow transductions with linear growth, and, as we will see, $\SRASP$ can define transductions with superlinear growth (\cref{ex:marked-squaring}).

\subsection{Properties}

\begin{lemma}
    \label{lemma:closure-under-composition-srasp}
    If $f_1\colon \Sigma_1^* \rightarrow \Sigma_2^*$ and $f_2\colon \Sigma_2^* \rightarrow \Sigma_3^*$ are computable in $\SRASP$, then their composition $f_2 \circ f_1 \colon \Sigma_1^* \rightarrow \Sigma_3^*$ is computable in $\SRASP$.
\end{lemma}
\begin{proof}
    Let the $\SRASP$ program $\progname{P}_i$ compute the transduction $f_i$ with minimum vector length $q_i$ for $i = 1,2$. Let $\progname{P}_3$ be the $\SRASP$ program that consists of the operations of $\progname{P}_1$ followed by the operations of $\progname{P}_2$, where $\progname{P}_1$ is modified to output a fresh vector $\vecname{z}$ (instead of $\outputy$) and $\progname{P}_2$ is modified to input vector $\vecname{z}$ (instead of $\inputx$).
    We can choose a nondecreasing 
    function
    $q_3$ such that $q_3(\ell) \ge \max(q_1(\ell), q_2(q_1(\ell)))$, so that $q_3$ as a minimum vector length ensures that 
    $\progname{P}_3$ correctly computes $f_2 \circ f_1$.
\end{proof}

\begin{lemma}
    \label{lemma:SRASP-computes-homomorphisms}
    For any string homomorphism $h: \Sigma^* \rightarrow \Gamma^*$ there exists an $\SRASP$ program to compute $h$, with minimum vector length $q(\ell) = K\ell$, where $K$ is the maximum of $|h(\sigma)|$ over $\sigma \in \Sigma$.
\end{lemma}

\begin{proof}
Number the symbols of $\Sigma$ as $\sigma_0, \ldots, \sigma_{m-1}$.
We use a position-wise operation to record in position $i$ the length of $h(\inputx(i))$.
\iffalse
\begin{raspcode}
\begin{align*}
\vecname{lens}(i) &= |h(\sigma_0)| \iftext \inputx(i) = \sigma_0 \\
&\phantom{={}} \quad \vdots \\
&\phantom{={}} \elsetext |h(\sigma_{m-1})| \iftext \inputx(i) = \sigma_{m-1} \\
&\phantom{={}} \elsetext 0
\end{align*}
\end{raspcode}
\fi
\begin{raspcode}
\begin{align*}
\vecname{lens}(i) = |h(\inputx(i))|
\end{align*}
\end{raspcode}
Then we determine the starting position of each $h(\inputx(i))$ in the output. 
\begin{raspcode}
\begin{align*}
\vecname{ends}(i) &= \attsum{j}{j \le i}{\vecname{lens}(j)} \\
\vecname{starts}(i) &= \vecname{ends}(i)-\vecname{lens}(i)
\end{align*}
\end{raspcode}
For $k \in [K]$, define $\vecname{sym}_k(i)$ such that if output position $i$ is to be the $k$-th symbol generated from input position $j$, then $\vecname{sym}_k(i) = \inputx(j)$:
\begin{raspcode}
\begin{align*}
\vecname{sym_0}(i) &= \attrdefault{j}{\true}{\pos(i)=\vecname{starts}(j)}{\inputx(j)}{\rechar{\pad}}\\
\vecname{sym_1}(i) &= \attrdefault{j}{j<i}{\true}{\vecname{sym_0}(j)}{\rechar{\pad}} \\
&\vdotswithin{=} \\
\vecname{sym_{K-1}}(i) &= \attrdefault{j}{j<i}{\true}{\vecname{sym_{K-2}}(j)}{\rechar{\pad}}
\end{align*}
\end{raspcode}
Finally, we can define the output vector:
\begin{raspcode}
\begin{align*}
\outputy(i) &= \sigma_0 \iftext \bigvee_{\substack{a \in \Sigma, k \in [K] \\ h(a)_k = \sigma_0}} \vecname{sym}_k(i) = a \\
&\phantom{={}} \quad \vdots \\
&\phantom{={}} \elsetext \sigma_{m-2} \iftext \bigvee_{\substack{a \in \Sigma, k \in [K] \\ h(a)_k = \sigma_{m-2}}} \vecname{sym}_k(i) = a \\
&\phantom{={}} \elsetext \sigma_{m-1}
\end{align*}
\end{raspcode}
An example is in \cref{ex:string_homomorphism}.
\end{proof}

\subsection{Examples and expressivity}

\begin{table}[t]
\begin{center}
\exampletablefontsize
\setlength{\tabcolsep}{5pt}
\begin{tabular}{@{}l|c|cccccc@{}} 
  \toprule
  \inputx           & input                 & A & B & B & C & $\pad$ & $\pad$ \\
  \pos              & $i$                   & $0$ & $1$ & $2$ & $3$ & $4$ & $5$ \\
  \vecname{lens}    & length of $h(\inputx(i))$    & $2$ & $0$ & $0$ & $3$ & $0$ & $0$ \\
  \vecname{ends}    & end of $h(\inputx(i))$    & $2$ & $2$ & $2$ & $5$ & $5$ & $5$ \\
  \vecname{starts}  & start of $h(\inputx(i))$     & $0$ & $2$ & $2$ & $2$ & $5$ & $5$ \\
  \vecname{sym0}    & mark start    & A & $\pad$ & C & $\pad$ & $\pad$ & $\pad$ \\
  \vecname{sym1}    & mark start+1    & $\pad$ & A & $\pad$ & C & $\pad$ & $\pad$ \\
  \vecname{sym2}    & mark start+2    & $\pad$ & $\pad$ & A & $\pad$ & C & $\pad$ \\
  \outputy    & output       & a & a & c & c & d & $\pad$ \\
  \bottomrule
 \end{tabular}
 \end{center}
 \caption{Example $\SRASP$ computation for a string homomorphism. Details in Ex.~\ref{ex:string_homomorphism}.}
 \label{tab:string_homomorphism}
\end{table}

\begin{example}[string homomorphisms]
\label{ex:string_homomorphism}
 Consider the homomorphism $\text{A} \mapsto \text{aa}$, $\text{B} \mapsto \varepsilon$, $\text{C} \mapsto \text{ccd}$.
     \begin{center}
         $\text{ABBC} \pad \pad \mapsto \text{aaccd} \pad$
     \end{center}
    \begin{raspcode}
    \begin{align*}
        \vecname{lens}(i) &= 2 \iftext \inputx(i)=\rechar{A} \\
                    &\hphantom{{}=} \elsetext 3 \iftext \inputx(i)=\rechar{C} \elsetext 0\\
        \vecname{ends}(i) &= \attsum{j}{j \le i}\vecname{lens}(j) \\
        \vecname{starts}(i) &= \vecname{ends}(i)-\vecname{lens}(i) \\
        \vecname{sym0}(i) &= \attrdefault{j}{\true}{\pos(i)=\vecname{starts}(j)}{\inputx(j)}{\rechar{\pad}}\\
        \vecname{sym1}(i) &= \attrdefault{j}{j<i}{\true}{\vecname{sym0}(j)}{\rechar{\pad}} \\
        \vecname{sym2}(i) &= \attrdefault{j}{j<i}{\true}{\vecname{sym1}(j)}{\rechar{\pad}} \\
        \outputy(i) &= 
                    \begin{aligned}[t]
                    &\rechar{a} \iftext \vecname{sym0}(i) = \rechar{A} \lor
                    \vecname{sym1}(i)=\rechar{A}\\ 
                    &  \elsetext \rechar{c} \iftext \vecname{sym0}(i)=\rechar{C} \lor
                    \vecname{sym1}(i)=\rechar{C}\\
                    & \elsetext \rechar{d} \iftext \vecname{sym2}(i)=\rechar{C}
                    \elsetext \rechar{\pad}
                    \end{aligned}
    \end{align*}
    \end{raspcode}
    An example run is in \cref{tab:string_homomorphism}.
\end{example}

\begin{example}[$\markedsquare$]
\label{ex:marked-squaring-rasp}
    Make $|w|$ many copies of $w$ separated by bars, with successively longer prefixes marked (here by uppercasing).
    \begin{center}
        $\text{abaa} \mapsto \text{\marker Abaa\marker  ABaa\marker ABAa\marker ABAA\marker}$
    \end{center}
This transduction is \FO{} polyregular but not regular. It has greater than linear growth, and is therefore not computable in $\BRASPpos$ with \compressed{} outputs. But it can be computed by the following $\SRASP$ program.
\begin{raspcode}
\begin{align*}
        \vecname{len}(i) &= \attl{j}{\true}{\inputx(j) = \rechar{\#}}{\pos(j)} \\
        \vecname{inpos}(i) &= \pos(i) < \vecname{len}(i) \\
        \vecname{glen}(i) &= \vecname{len}(i) + 1 \iftext \pos(i) > 0 \elsetext 0\\
        \vecname{mglen}(i) &= \attsum{j}{j \le i}{\vecname{glen}(j)} \\
        \vecname{starts}(i) &= \vecname{mglen}(i) \iftext \vecname{inpos}(i) \elsetext 0 \\
        \vecname{isstart}(i) &= \attldefault{j}{\true}{\pos(i) = \vecname{starts}(j)}{\true}{\false}\\
        \vecname{isstartnum}(i) &= 1 \iftext \vecname{isstart}(i) \elsetext 0 \\
        \vecname{gnumber}(i) &= \attsum{j}{j \le i}{\vecname{isstartnum}(j)} \\
        \vecname{gstart}(i) &= \attr{j}{j \le i}{\vecname{isstart}(j)} \pos(j) \\
        \vecname{src}(i) &= \pos(i) - \vecname{gstart}(i) - 1 \\
        \vecname{ismarked}(i) &= \vecname{src}(i) < \vecname{gnumber}(i) \\
        \vecname{y1}(i) &= \inputx(\vecname{src}(i))\\
        \vecname{y2}(i) &= \rechar{|} \iftext \vecname{isstart}(i) \\
        &\phantom{={}} \elsetext \mathit{mark}(\vecname{y1}(i)) \iftext \vecname{ismarked}(i) \\
        &\phantom{={}} \elsetext \vecname{y1}(i) \\
        \vecname{lastbar}(i) &= \attl{j}{\true}{\vecname{y2}(j) = \rechar{\pad}}{\pos(j)} \\
        \outputy(i) &= \rechar{|} \iftext \pos(i) = \vecname{lastbar}(i) \elsetext \vecname{y2}(i)
\end{align*}
\end{raspcode}
The finite function $\mathit{mark}$ changes the input symbol to uppercase.
An example run is in \cref{tab:marked-squaring}.
\end{example}
\begin{table*}[ht]
\begin{center}
\exampletablefontsize
\setlength{\tabcolsep}{5pt}
\begin{tabular}{@{}l|c|cccccccccccccc@{}} 
  \toprule
  \inputx           & input                             & a & a & b & $\pad$ & $\pad$ & $\pad$ & $\pad$ & $\pad$ & $\pad$ & $\pad$ & $\pad$ & $\pad$ & $\pad$ & $\pad$ \\
  \pos              & $i$                               & $0$ & $1$ & $2$ & $3$ & $4$ & $5$ & $6$ & $7$ & $8$ & $9$ & $10$ & $11$ & $12$ & $13$  \\
  \vecname{len}     & input length                      & $3$ & $3$ & $3$ & $3$ & $3$ & $3$ & $3$ & $3$ & $3$ & $3$ & $3$ & $3$ & $3$ & $3$  \\
  \vecname{inpos}   & $i$ in input?                        & $\top$ & $\top$ & $\top$ & $\bot$ & $\bot$ & $\bot$ & $\bot$ & $\bot$ & $\bot$ & $\bot$ & $\bot$ & $\bot$ & $\bot$ & $\bot$  \\
  \vecname{glen}     & $\ell_g$ = group length ($i > 0$)  & $0$ & $4$ & $4$ & $4$ & $4$ & $4$ & $4$ & $4$ & $4$ & $4$ & $4$ & $4$ & $4$ & $4$ \\
  \vecname{mglen}    & $\min(n-1,i \ell_{\vecname{g}})$            & $0$ & $4$ & $8$ & $12$ & $13$ & $13$ & $13$ & $13$ & $13$ & $13$ & $13$ & $13$ & $13$ & $13$  \\
  \vecname{starts}  & starts of groups                  & $0$ & $4$ & $8$ & $0$ & $0$ & $0$ & $0$ & $0$ & $0$ & $0$ & $0$ & $0$ & $0$ & $0$ \\
  \vecname{isstart} & is $i$ in starts?                 & $\top$ & $\bot$ & $\bot$ & $\bot$ & $\top$ & $\bot$ & $\bot$ & $\bot$ & $\top$ & $\bot$ & $\bot$ & $\bot$ & $\bot$ & $\bot$ \\
  \vecname{isstartnum} & \vecname{isstart} numeric      & $1$ & $0$ & $0$ & $0$ & $1$ & $0$ & $0$ & $0$ & $1$ & $0$ & $0$ & $0$ & $0$ & $0$ \\
  \vecname{gnumber} & group number    & $1$ & $1$ & $1$ & $1$ & $2$ & $2$ & $2$ & $2$ & $3$ & $3$ & $3$ & $3$ & $3$ & $3$ \\
  \vecname{gstart}  & start of $i$'s group              & $0$ & $0$ & $0$ & $0$ & $4$ & $4$ & $4$ & $4$ & $8$ & $8$ & $8$ & $8$ & $8$ & $8$ \\
  \vecname{src}      & $i - \vecname{gstart}(i) - 1$   & $0$ & $0$ & $1$ & $2$ & $0$ & $0$ & $1$ & $2$ & $0$ & $0$ & $1$ & $2$ & $3$ & $4$ \\
  \vecname{ismarked}  & is $i$ marked?                         &  $\top$ &  $\top$ & $\bot$ & $\bot$ &  $\top$ &  $\top$ & $\top$ & $\bot$ &  $\top$ &  $\top$ & $\top$ & $\top$  & $\bot$ & $\bot$ \\
  \vecname{y1}      &  letters moved             & a & a & a & b & a & a & a & b & a & a & a & b & $\pad$ & $\pad$ \\
  \vecname{y2}      &   mark and add initial $\rechar{|}$'s    & | & A & a & b & | & A & A & b & | & A & A & B & $\pad$ & $\pad$  \\
  \vecname{lastbar} &    $i$ for last $\rechar{|}$      & $12$ & $12$ & $12$ & $12$ & $12$ & $12$ & $12$ & $12$ & $12$ & $12$ & $12$ & $12$ & $12$ & $12$ \\
    \outputy      &  output     & | & A & a & b & | & A & A & b & | & A & A & B & | & $\pad$ \\
  \bottomrule
 \end{tabular}
 \end{center}
 \caption{Example $\SRASP$ computation for $\markedsquare$. Details in Ex.~\ref{ex:marked-squaring-rasp}.}
 \label{tab:marked-squaring}
\end{table*}

\begin{theorem}
    \label{theorem:a-rasp-includes-fo-polyregular}
    Every \FO{} polyregular transduction is computable in $\SRASP$.
\end{theorem}
\begin{proof}
By \cref{def:aperiodic-transduction-classes}, any \FO{} polyregular transduction can be decomposed into a composition of \FO{} regular transductions and $\markedsquare$. All \FO{} regular transductions are computable in $\BRASPpos$ (\cref{thm:aregular_to_brasppos}), and their \compressed{} outputs can be unpacked in $\SRASP$ (\cref{lemma:SRASP-computes-homomorphisms}), so all \FO{} regular transductions are computable in $\SRASP$. Further, $\markedsquare$ is computable in $\SRASP$ (\cref{ex:marked-squaring-rasp}) and \SRASP{} is closed under composition (\cref{lemma:closure-under-composition-srasp}).
Thus, \SRASP{} can compute all \FO{} polyregular transductions.
\end{proof}

\begin{example}[$\majorityrules$]
\label{ex:majorityrules}
\begin{samepage}
    If there are at least as many $a$'s as $b$'s in the input, change all inputs to $a$; otherwise change inputs to $b$ \citep{bakovic2000thesis}.
    \begin{center}
    $\text{abbabbba} \pad \pad \mapsto \text{bbbbbbbb} \pad \pad $
    \end{center}
\end{samepage}    
    The number of $a$'s and the number of $b$'s are computed and broadcast
    to every position. Each position determines whether its output is
    $a$, $b$ or $\pad$.
    \begin{raspcode}
    \begin{align*}
        \vecname{pa}(i) &= \attsum{j}{j \le i}{(1 \iftext \inputx(j)=\rechar{a} \elsetext 0)} \\
        \vecname{na}(i) &= \attr{j}{\true}{\true}{\vecname{pa}(j)}\\
        \vecname{pb}(i) &= \attsum{j}{j \le i}{(1 \iftext \inputx(j)=\rechar{b} \elsetext 0)} \\
        \vecname{nb}(i) &= \attr{j}{\true}{\true}{\vecname{pb}(j)}\\
        \outputy(i) &= \pad \iftext \inputx(i) = \pad \\
        & \hphantom{{}=} \elsetext \rechar{a} \iftext \vecname{na}(i) \ge \vecname{nb}(i) \elsetext \rechar{b}
    \end{align*}
    \end{raspcode}
    An example run is in \cref{tab:majority-rules}.
\end{example}
\begin{table}
\begin{center}
\exampletablefontsize
\setlength{\tabcolsep}{4.5pt}
\begin{tabular}{@{}l|c|cccccccccc@{}}
  \toprule
  \inputx       & input         & b & b & a & b & b & a & b & a & $\pad$ & $\pad$\\
  \vecname{pos}     & $i$           & $0$ & $1$ & $2$ & $3$ & $4$ & $5$ & $6$ & $7$ & $8$ & $9$\\
  \vecname{pa}      & count-left(a) & $0$ & $0$ & $1$ & $1$ & $1$ & $2$ & $2$ & $3$ & $3$ & $3$ \\
  \vecname{na}      & count(a)      & $3$ & $3$ & $3$ & $3$ & $3$ & $3$ & $3$ & $3$ & $3$ & $3$\\
  \vecname{pb}      & count-left(b) & $1$ & $2$ & $2$ & $3$ & $4$ & $4$ & $5$ & $5$ & $5$ & $5$ \\
  \vecname{nb}      & count(b)      & $5$ & $5$ & $5$ & $5$ & $5$ & $5$ & $5$ & $5$ & $5$ & $5$\\
  \outputy       &  output & b & b & b & b & b & b & b & b & $\pad$ & $\pad$\\
  \bottomrule
 \end{tabular}
 \end{center}
 \caption{Example $\SRASP$ computation for $\majorityrules$.  Details in Ex.~\ref{ex:majorityrules}.}
 \label{tab:majority-rules}
\end{table}

\begin{proposition}
\label{thm:majority-rules-not-in-polyregular}
The transduction $\majorityrules$ is neither polyregular nor computable in $\BRASPpos$.
\end{proposition}
\begin{proof}
   Polyregular transductions preserve regular languages under inverse \citep[Thm.~1.7]{bojanczyk2018}.
   The preimage of the regular language $a^*$ under $\majorityrules$ is
    $M = \{w \mid \text{$w$ contains more $\text{a}$'s than $\text{b}$'s} \}$, which is not regular, so $\majorityrules$ is not polyregular.
   
   A circuit family computing $\majorityrules$ can be modified to decide $M$, which is not in $\AC^0$ \citep{furst+:1984}. Thus the $\majorityrules$ transduction is not computable in $\BRASPpos$.
   \end{proof}

\begin{example}[\modularcountermodm]
Let $m$ be a positive integer and define the transduction \modularcountermodm{} to map any input sequence $a_0a_1 \cdots a_{n-1}$ to the sequence $b_0b_1 \cdots b_{n-1}$ where $b_i = (\sum_{0}^i a_j) \bmod m$.
This transduction is rational but not aperiodic; it is a generalization of the parity problem, which has been discussed at length elsewhere \citep{hahn-2020-theoretical,chiang-cholak-2022-overcoming}.
\end{example}

\begin{proposition}
\label{prop:SRASP-modular-counter-mod-m}
For any $m$, \SRASP{} can compute the transduction \modularcountermodm{}.
\end{proposition}
\begin{proof}
We just give the case of $m = 3$, which is easily generalized.
The vector \vecname{residues} contains the residues of positions modulo $3$ computed by the program in \cref{prop:BRASPpos-computes-residues}.
Define the finite function $\mathit{fmod3}(x,y) = (x+2y) \bmod 3$ for $x,y \in [3]$.
\begin{raspcode}
\begin{align*}
    \vecname{ones}(i) &= 1 \iftext \inputx(i) = 1  \elsetext 0\\
    \vecname{ps1}(i) &= \attsum{j}{j \le i}{\vecname{ones}(j)} \\
    \vecname{ps1m3}(i) &= \vecname{residues}(\vecname{ps1}(i)) \\
    \vecname{twos}(i) &= 1 \iftext \inputx(i) = 2 \elsetext 0\\
    \vecname{ps2}(i) &= \attsum{j}{j \le i}{\vecname{twos}(j)} \\
    \vecname{ps2m3}(i) &= \vecname{residues}(\vecname{ps2(i))} \\
    \outputy(i) &= \mathit{fmod3}(\vecname{ps1m3}(i),\vecname{ps2m3}(i)) \tag*{\qedhere}
\end{align*}
\end{raspcode}
\end{proof}

On the other hand, because prefix sum can be simulated by a family of $\TC^0$ circuits (threshold circuits of constant depth and polynomial size), any transduction computable in $\SRASP$ is in $\TC^0$.

\subsection{Average-hard attention transformers}
\label{sec:S-RASP-with-augmented-average-hard-attention}

We prove the following connection between $\SRASP$ programs and average hard attention transformers in \cref{sec:sim-SRASP-proof}.
\begin{theorem} 
\label{thm:srasp_to_ahat}
Any transduction computable by an $\SRASP$ program is computable by a masked average-hard attention transformer encoder with a position encoding of $i/n$, $(i/n)^2$, and $1/(i+2)$. 
\end{theorem}

One consequence is the following result relating unique-hard and average-hard attention:
\begin{corollary}
\label{thm:uhat_to_ahat}
Any transduction computable by a masked unique-hard attention transformer encoder can be computed by a masked average-hard attention transformer encoder with a position encoding of $i/n$, $(i/n)^2$, and $1/(i+2)$.
\end{corollary}

\section{Conclusions}

This is, to our knowledge, the first formal study of transformers for sequence-to-sequence transductions,
using variants of RASP to connect classes of transformers to classes of transductions.
We showed that unique-hard attention transformers and $\BRASP$ compute precisely the class of \FO{} rational transductions; $\BRASPpos$ strictly contains all \FO{} regular transductions; and average-hard attention transformers and $\SRASP$ strictly contain all \FO{} polyregular transductions.
Our finding that $\BRASPpos$ and $\SRASP$ can compute transductions outside the corresponding \FO{} class in the transduction hierarchy raises the question of fully characterizing their expressivity, a promising future research direction.

\section*{Acknowledgements}

We thank Miko{\l}aj Boja{\'n}czyk, Micha{\"e}l Cadilhac, L{\^e} Th{\`a}nh D{\~u}ng (Tito) Nguy{\~{\^{e}}}n, and the anonymous reviewers for their very helpful advice.

\bibliographystyle{acl_natbib}
\bibliography{SRASP}

\clearpage
\appendix

\section*{Appendices}

In the following appendices, we prove \cref{thm:srasp_to_ahat}.
\Cref{sec:ahat} reviews the definition of average-hard attention transformers.
\Cref{sec:sim-SRASP-proof} contains our main proof, while
\cref{sec:alternate_pe} contains another construction using a different position embedding. \Cref{sec:other_approaches} compares some features of our simulation with other simulations.

\section{Average hard attention transformers}
\label{sec:ahat}

We recall the definition of a transformer encoder with average-hard attention, also known as saturated attention \citep{yao2023selfattention,hao-etal-2022-circuits,barcelo2023logical}.
Let $d > 0$ and $n \ge 0$. An \emph{activation sequence} is a sequence of $n$ vectors in $\R^d$, one for each string position.
The positions are numbered $-1, 0, 1, \ldots, n-1$. Position $-1$, called the \emph{default position}, does not hold an input symbol and will be explained below.
A transformer encoder is the composition of a constant number (independent of $n$) of layers, which of which maps an activation sequence $u_{-1},\ldots,u_{n-1}$ to an activation sequence $u_{-1}',\ldots,u_{n-1}'$.

There are two types of layers: (1) position-wise and (2) average hard attention.
A \emph{position-wise layer} computes a function $u_i' = u_i + f(u_i)$ for all positions $i$, where $f$ is %
a position-wise two-layer feed-forward network (FFN) with ReLU activations. An \emph{average hard attention layer} is specified by three linear transformations $Q,K,V \colon \R^d \to \R^d$.
The dot product $S(i,j) = \langle Qu_i,Ku_j \rangle$ is the \emph{attention score} from position $i$ to position $j$.
For each position~$i$, let $M_i$ be the set of positions $j$ that maximize $S(i,j)$.
Then $u'_i = u_i + (\sum_{j \in M_i} Vu_j)/|M_i|$.
An average hard attention layer may be \emph{masked} using \emph{strict} or \emph{non-strict future masking}, in which for each position $i$, only positions $j < i$ or $j \le i$ (respectively) are considered in the attention calculation.
With strict future masking, the default position has nowhere to attend to, so the result is $u_{-1}' = u_{-1}$.

\section{Simulating $\SRASP$}
\label{sec:sim-SRASP-proof}

\subsection{Overview of the simulation}

To define the computation of a transduction by a transformer, we need to specify how the input and output strings are represented in the initial and final activation sequences.
If the input string is $w = a_0 a_1 \cdots a_{\ell-1}$, we let $a_i = \pad$ for $i = \ell, \ldots, n-1$.

Let $\Sigma \cup \{\pad\} = \{\sigma_0,\ldots,\sigma_{k-1}\}$ be totally ordered.
The first $k$ coordinates of each activation vector hold the one-hot encoding of $a_i$
(or the zero vector at the default position).
The representation of the output string is analogous, using the alphabet $\Gamma \cup \{\pad\}$.

Five more coordinates are designated to hold the \emph{position encoding} (PE) and 
quantities computed from it.  Descriptive names for these coordinates of the activation vector at position $i$ are as follows. 
\[
\exampletablefontsize
\begin{blockarray}{r[c]}
0 & \mathbb{I}[a_i = \sigma_0] \\
\vdots & \vdots \\
k-1 & \mathbb{I}[a_i = \sigma_{k-1}] \\
\mathit{pos} & i/n \\
\mathit{posq} & (i/n)^2 \\
\mathit{posi} & 1/(i+2) \\
\mathit{default} & \mathbb{I}[i=-1] \\
\mathit{zero} & \mathbb{I}[i=0] \bigstrut[b]
\end{blockarray}
\]
where $\mathbb{I}[\mathord\cdot]$ is $1$ if the argument is true, $0$ otherwise.
In the simulation, $i/n$ is used for sum and difference, $i/n$ and $(i/n)^2$ are used for equality comparison, and $i/n$, $(i/n)^2$ and $1/(i+2)$ are used for the prefix sum operation.
We note that the last two coordinates above can be computed from $i/n$.

We turn to how $\SRASP$ programs may be simulated.
Vectors of Boolean, symbol, and integer values in an $\SRASP$ program are represented in one or more coordinates of an activation sequence in the transformer.
Each operation of an $\SRASP$ program computes a new vector of values, and is simulated by one or more transformer encoder layers which compute new values in one or more coordinates of the activation sequence.
Assume that $\progname{P}_f$ is an $\SRASP$ program computing a transduction $f:\Sigma^* \rightarrow \Gamma^*$ with minimum vector length $q(\ell)$,
and that $n > q(\ell)$.

\subsection{Representing $\SRASP$ vectors}

Vectors of Booleans, symbols, and integers in the program $\progname{P}_f$ are represented in the activation sequence of the transformer as follows.

Each Boolean vector $v_0,v_1,\ldots,v_{n-1}$ in $\progname{P}_f$ is represented by one coordinate $r$ of the transformer activation sequence $u_{-1},u_0,\ldots,u_{n-1}$, where for each $i \in [n]$, $u_i[r] = 0$ if $v_i = \false$ and $u_i[r] = 1$ if $v_i = \true$.
For the default position, $u_{-1}[r] = 0$.

Let $\Delta = \{\delta_0,\delta_1,\ldots,\delta_k\}$ denote the finite set of all symbols that appear in any symbol vector in~$\progname{P}_f$.
Each symbol vector $v_0,v_1,\ldots,v_{n-1}$ in $\progname{P}_f$ is represented by $|\Delta|$ coordinates $r_0,r_1,\ldots,r_{k-1}$,
which hold a one-hot representation of $v_i$ (or the zero vector at the default position).

Each integer vector $v_0,v_1,\ldots,v_{n-1}$ in the program is represented by a specified coordinate $r$ in the transformer activation sequence, where for each $i \in [n]$, $u_i[r] = v_i/n$.
In the PE, the value of $u_{-1}[\mathit{pos}]$ is $-1/n$, but for other integer vectors we have $u_{-1}[r] = 0$.
We note that all of the representing values are less than or equal to $1$.

\subsection{Table lookup}

A key property of $\SRASP$ is that every integer value computed in the program must be equal to some position index $i \in [n]$. 
We use this property to implement a table lookup operation.

\begin{lemma} \label{thm:quadratic_maximization}
    For any integers $x$, $q$, let
    \[f_q(x) = 2qx - x^2.\] 
    Then:
    \begin{compactenum}
        \item $f_q(x)$ is uniquely maximized at $x = q$;
        \item if $x \ne q$, then $f_q(q) - f_q(x) \geq 1$.
    \end{compactenum}
\end{lemma}

\begin{proof} 
This is a generalized version of a technique by \citet{barcelo2023logical}.
It can easily be shown by looking at the first and second derivatives of $f$, and by comparing $f_q(q)$ with $f_q(q-1)$ and $f_q(q+1)$.
\end{proof}

\begin{lemma}
    \label{lemma:copy-from-basic}
    Fix an activation sequence $u_{-1},\ldots,u_{n-1}$ and coordinates $r,s,t$ such that $u_i[r] = k_i/n$, where each $k_i \in [n]$.
    Then there is an average-hard attention layer that computes $u'_{-1}, \ldots, u'_{n-1}$, where $u'_i[t] = u_{k_i}[s]$ and the other coordinates stay the same.
\end{lemma}

\begin{proof}
    Consider an attention layer with no mask and the following attention score:
    \begin{align*}
    S(i,j) &= 2u_i[r] u_j[\mathit{pos}] - u_j[\mathit{posq}] \\
    &= \frac{2k_i j - j^2}{n^2}.
    \end{align*}
    which is a bilinear form in $u_i$ and $u_j$, and (by \cref{thm:quadratic_maximization}) is uniquely maximized when $j = k_i$.
    The value is $u_j[s]$,
    which is stored in coordinate $t$ of the output activation sequence. %
\end{proof}

We remark that if $k_i \ge n$, the unique maximizing value of $S(i,j)$ for $j \in [-1,n-1]$ is $j = n-1$, so the attention layer in the proof above returns the value $v_{n-1}$ for such positions $i$.

\subsection{Simulating $\SRASP$ operations}

For each operation below, let $u_{-1},\ldots,u_{n-1}$ be the input activation sequence, and let $u'_{-1},\ldots,u'_{n-1}$ be the output activation sequence. If $k$, $v_1$, $v_2$, $b$, and $t$ are $\SRASP$ vectors, we also write $k$, $v_1$, $v_2$, $b$, and $t$, respectively, for the coordinates representing them in the transformer. 

\subsubsection{Position-wise operations}

Position-wise Boolean operations on Boolean vectors can be simulated exactly by position-wise FFNs, as shown by \citet{yang2024masked}.
Position-wise operations on symbol values reduce to Boolean operations on Boolean values.

To simulate \emph{addition of two integer vectors}, $t(i) = v_1(i) + v_2(i)$,
we first use a FFN to compute $k/n = \max(0, u_i[v_1]+u_i[v_2])$.
The result may exceed $(n-1)/n$, so we use table lookup (\cref{lemma:copy-from-basic}) to map $k/n$ to $u_k[\mathit{pos]}$; 
this sets values larger than $(n-1)/n$ to $(n-1)/n$. The result is stored in $u'_i[t]$. Subtraction is similar, with ReLU ensuring the result is non-negative.

For \emph{position-wise comparison of integer vectors} $t(i) = v_1(i) \le v_2(i)$,  %
we use a FFN to compute $k/n = \max(0, u_i[v_1] - u_i[v_2])$.
We use table lookup to map $k/n$ to $u_k[\mathit{zero}]$, which is $1$ if $u_i[v_1]-u_i[v_2]\le 0$, and $0$ otherwise. 
The other comparison operators are similar.

For the \emph{position-wise operation $t(i) = v_1(i) \iftext b(i) \elsetext v_2(i)$}:  %
If $v_1$ and $v_2$ are both Boolean vectors or both symbol vectors, this can be reduced to position-wise Boolean operations.
If $v_1$ and $v_2$ are integer vectors, %
we use a FFN to compute
\begin{align*}
u_i'[t] &= \max(0,u_i[v_1]+u_i[b]-1)\\
&\quad {} + \max(0,u_i[v_2]-u_i[b]).
\end{align*}
Thus if $u_i[b] = 1$ then $u_i'[t]=u_i[v_1]$ and $u_i'[t] = 0$, and if $u_i[b]=0$ then $u_i'[t]=0$ and $u_i'[t]=u_i[v_2]$.

\subsubsection{Prefix sum}
\label{sec:prefix-sum}

Next, we turn to the \emph{prefix sum} operation, $t(i) = \attsum{j}{j \le i}{k(j)}$.
Assume that $u_i[k] = k(i)/n$, where each $k(i)$ is an integer in $[n]$ and $k(-1) = 0$. Let $p_i \ge 0$ be the sum of $k(-1),k(0),\ldots,k(i)$ and let $p'_i = \min(n-1,p_i)$, 
which is the sequence of values to be computed and stored in coordinate $t$.

The first attention layer uses non-strict future masked average hard attention with $S(i,j) = 0$, and the value is $u_j[k]$. %
The resulting activation sequence has the following values in coordinate $s$: 
\begin{equation} \label{eq:prefix_sum_result}
\frac{0}{n},\frac{p_0}{2n},\frac{p_1}{3n},\ldots,\frac{p_{n-1}}{(n+1)n}.
\end{equation}
Each value is smaller than the desired value by a factor of $(i+2)$; to remove this factor, we use a second attention layer.
Let $v_{-1},v_0,\ldots,v_{n-1}$ denote the activation sequence after the first layer.
We use an average hard attention layer with no mask and the following attention score:
\begin{align*}
S(i,j) &= 2v_i[s] v_j[\mathit{pos}] - v_i[\mathit{posi}] v_j[\mathit{posq}] \\
&= \frac{2p_i j - j^2}{(i+2)n^2}
\end{align*}
which is a bilinear form in $v_i$ and $v_j$, and (by \cref{thm:quadratic_maximization}) is uniquely maximized when $j = p_i$. As in the remark after Lemma~\ref{lemma:copy-from-basic}, if $p_i \ge n$, the maximizing $j \in [n]$ is $j = n-1$.
The value is $v_j[\mathit{pos}] = j/n = p'_{i-1}/n$ for $i>0$ (and $0$ if $i=0$), which is assigned to coordinate $t$, and the other coordinates are unchanged.

\subsubsection{Leftmost and rightmost attention}
\label{sec:leftmost-rightmost-attention}
The operations
\begin{gather*}
t(i) = \attldefault{j}{M(i,j)}{S(i,j)}{V(j)}{D(i)} \\
t(i) = \attrdefault{j}{M(i,j)}{S(i,j)}{V(j)}{D(i)}
\end{gather*}
require that if there is any position $j \in [n]$ that makes the attention predicate $S(i,j)$ true, then the unique minimum or maximum such $j$ is selected, but if there is no satisfying position $j \in [n]$, then the default value is used.
Attention may be past or future masked, either strictly or non-strictly.
We assume that transformers have only (strict or non-strict) future masking; to simulate past masking, we can calculate the index $(n-1)/n - i/n$, use \cref{lemma:copy-from-basic} to reverse the relevant vectors, and then use future masking.

The attention score $S(i,j)$ is either a Boolean combination of Boolean vectors, or an equality comparison between two integer vectors.
In either case, we compute an attention score \[S'(i,j) = S_{\text{base}}(i,j) \pm S_{\text{tie}}(i,j) + S_{\text{def}}(i) \,\mathit{default}(j)\]
where $S_{\text{base}}(i,j)$ is maximized for positions where $S(i,j)$ is true, $+S_{\text{tie}}$ breaks ties to right, $- S_{\text{tie}}$ to the left, and $S_{\text{def}}$ handles the default case.

\paragraph{Maximization.}
If $S(i,j)$ is a Boolean combination of Boolean vectors,
to ensure that attention from any position to the default position is $0$, we let $S_{\text{base}}(i,j) = \neg \mathit{default}(j) \land S(i,j)$.
This may be computed by dot product attention, as described by \citet{yang2024masked}.

For the special case where $S(i,j)$ is an equality comparison of integer vectors, say $v_1(i)=v_2(j)$:
We first use a lookup operation (\cref{lemma:copy-from-basic}) with the $\mathit{posq}$ entry of the PE to get the squares of the values in $v_2$ in coordinate $t$.
Let $u_{-1},u_0,\ldots,u_{n-1}$ be the resulting activation sequence.
We then use an average hard attention operation with the attention score function
\begin{align*}
S_{\text{base}}(i,j) &= 2u_i[v_1] u_j[v_2] - u_j[t] \\
&= \frac{2v_1(i) v_2(j) - v_2(j)^2}{n^2}
\end{align*}
which is a bilinear form in $u_i$ and $u_j$, and is maximized (by \cref{thm:quadratic_maximization}) when $v_2(j) = v_1(i)$.

\paragraph{Breaking ties.}

If $S_{\text{base}}(i,j)$ were used with average hard attention, then the activation values would be averaged for all the satisfying $j$.
To ensure that the maximum satisfying position $j$ has a unique maximum score, we break ties by adding or subtracting $S_{\text{tie}}(i,j)$.
We must ensure that the values added or subtracted are smaller than the minimum difference between the values for satisfying and non-satisfying positions.

For a Boolean combination of Boolean vectors, let
$S_{\text{tie}}(i,j) = \max(0, j/(2n))$.
Then under rightmost attention, the rightmost satisfying $j$ has the highest attention score, which is at least $1$, while every non-satisfying~$j$ has an attention score less than $1/2$. Similarly for leftmost attention.

For an equality comparison $v_1(i) = v_2(j)$,
the difference between the maximum score attained and any other score is at least $(1/n)^2$ by \cref{thm:quadratic_maximization}. So if we add or subtract values less than $(1/n)^2$, no non-equality score can exceed an equality score.
This can be achieved by letting $S_{\text{tie}}(i,j) = j/(2n^3)$.
This is computable using dot product attention because $j/n$ is in the PE for $j$ and $(1/n)^2$ is in the PE for $1$ and can be initially broadcast to all positions.

\paragraph{Default values.}
The term $S_{\text{def}}$ needs to give the default position an attention score strictly between the possible scores for satisfying and non-satisfying~$j$.

For a Boolean combination of Boolean vectors, the maximum non-satisfying score is less than $1/2$ and the minimum satisfying score is at least $1$, so if we let $S_{\text{def}}(i) = 3/4$, then the default position has an attention score of $3/4$, so it will be the unique maximum in case there are no satisfying positions.

For an equality comparison of integer vectors, the maximum non-satisfying score is less than $(v_1(i)/n)^2 - (1/2)(1/n^2)$, and the minimum satisfying score is at least $(v_1(i)/n)^2$, so $S_{\text{def}}(i) = (v_1(i)/n)^2-(1/4)(1/n^2)$ is strictly between these values.
The value of $(v_1(i)/n)^2$ may be obtained at position $i$ using \cref{lemma:copy-from-basic} with index $v_1(i)/n$ and the $\textit{posq}$ coordinate of the PE. 

Thus, the default position is selected when there is no $j \in [n]$ satisfying the attention predicate; it remains to supply the default value.
We use an attention layer with the attention score $S'$ given above and value $V(j) = \begin{bsmallmatrix} \mathit{default}(j) \\ V(j) \end{bsmallmatrix}$.
Let $j_i$ be the position that $i$ attends to.
Then we use a position-wise if/else operation that returns (the simulation of) $D(i)$ if $\mathit{default}(j_i) = 1$ and $V(j_i)$ otherwise. 
This concludes the proof of Theorem~\ref{thm:srasp_to_ahat}.

\section{An Alternate Position Encoding}
\label{sec:alternate_pe}

The simulation of \SRASP\ via average hard attention transformers in \cref{thm:srasp_to_ahat} relies on three kinds of position encoding: $i/n$, $(i/n)^2$, and $1/(i+2)$. In this section, we present evidence for the following.
\begin{conjecture}
    \label{conj:reduced-PE}
    Any transduction computable by an $\SRASP$ program is computable by a masked average-hard attention transformer encoder with a position encoding of $i/n$.
\end{conjecture}

First, $1/(i+2)$ can be computed from $i/n$.
\begin{proposition} \label{prop:one_over_i}
    A transformer with positions $i \in \{-1, 0, \ldots, n\}$ and position encoding $i/n$ can compute $1/(i+2)$ at all positions $i$.
\end{proposition}
\begin{proof}

As observed by \citet{merrill2024expressive}, a transformer can use the $i/n$ encoding to uniquely identify the first position ($-1$) and compute $1/(i+2)$ by using non-strict future masked attention with value $1$ at that position and $0$ elsewhere ($0, \ldots, n-1$).
\end{proof}In \cref{thm:srasp_to_ahat_inverse_sq_emb_v1} we show that the position encoding $i/n$ and $1/(i+2)^2$ suffices for the simulation of \SRASP{} by a masked average hard attention transformer.
Though it's unclear whether a transformer with position encoding $i/n$ can compute $1/(i+2)^2$,
we note the following. 
\begin{proposition}
    \label{prop:derived-position-encodings}
    A transformer with positions $i \in \{-1, 0, \ldots, n\}$ and position encoding $i/n$ can compute $1/((i+2)^2 - 1)$ at positions $i < n$.
\end{proposition}
\begin{proof}
By \cref{prop:one_over_i}, the transformer can compute $1/(i+2)$ at position $i$.
It can then compute $1/((i+2)^2-1)$ simply as the difference between the $1/(i+2)$ values at the two neighbors of position $i$:
\[
    \frac{1}{(i+2)^2 - 1} = \frac{1}{2} \left( \frac{1}{i+1} - \frac{1}{i+3} \right).
    \qedhere
\]
\end{proof}

\begin{theorem} 
\label{thm:srasp_to_ahat_inverse_sq_emb_v1}
Any transduction computable by an $\SRASP$ program is computable by a masked average-hard attention transformer encoder with a position encoding of $i/n$ and $1/(i+2)^2$.
\end{theorem}

\begin{proof}[Proof Sketch]
The proof of this theorem closely follows the argument presented earlier for \Cref{thm:srasp_to_ahat}, except for the position encoding used. We will show how each use of $(i/n)^2$ in that original argument can be replaced with an equivalent use of $1/(i+2)^2$, which we assume to be stored in a coordinate called $\mathit{posiq}$ (for ``\underline{i}nverse \underline{q}uadratic'').
We also assume that $1/(i+2)$ is available by \cref{prop:one_over_i}.

The original proof uses the quadratic maximization in \cref{thm:quadratic_maximization}, which we replace with:
\begin{lemma} \label{thm:quadratic_maximization_alt}
    For any integers $x$, $q$, let
    \begin{equation}
    \label{eqn:inverse-sq-fractional}
    f_q(x) = \frac{2}{n(x+2)} - \frac{q+2}{n(x+2)^2}.
    \end{equation}
    Then
        $f_q(x)$ is uniquely maximized over values of $x \ge -1$ when $x = q$.
\end{lemma}
\begin{proof}
Consider the derivative, $-2/n(j+2)^2 + 2(q+2)/n(j+2)^3$, whose only real-valued root is $j = q$. Furthermore, the derivative is positive for $j < q$ and negative for $j > q$.
\end{proof}
This score is a bilinear form that can be computed via average hard attention using query $\langle 2/n, - q/n - 2/n \rangle$ at position $i$ and key $\langle 1/(j+2), 1/(j+2)^2 \rangle $ at position $j$. %
In all our applications of this new score, we will ensure that $q/n$ is available at position $i$. The $2/n$ term can also be computed at position $i$
by attending uniformly (without masking) with value $2$ at the first position and $0$ elsewhere. There are three uses of $\mathit{posq}$ in the original argument that we have to modify.

The first use is in the proof of \Cref{lemma:copy-from-basic}, for the basic lookup operation. Instead of using an attention score of $2u_i[r]  u_j[\mathit{pos}] - u_j[\mathit{posq}]$, we use \Cref{eqn:inverse-sq-fractional} with $q = k_i$ (recall that $u_i[r] = k_i / n$):
\begin{equation*}
    S(i, j) = \frac{2}{n} \, u_j[\mathit{posi}] - \left(u_i[r] + \frac{2}{n}\right) u_j[\mathit{posiq}] .
\end{equation*}
By \cref{thm:quadratic_maximization_alt}, $S(i, j)$ is maximized over $j$ uniquely when $j = k_i$, as needed in the proof of \cref{lemma:copy-from-basic}. %

The next use of $\mathit{posq}$ in the original argument is for the prefix sum (\cref{sec:prefix-sum}). As before, we compute $p_i / ((i+2)n)$ and store it as $v_i[s]$, with $v_i[0]$ being $0/n$. Instead of the original attention score of $2v_i[s] v_j[\mathit{pos}] - v_j[\mathit{posq}] v_i[\mathit{posi}]$, we use:
\begin{align*}
    S(i, j) & = \frac{2}{n} \, v_j[\mathit{posi}] v_i[\mathit{posi}] - {} \\
      & \qquad \left( v_i[s] + \frac{2}{n(i+2)} \right) v_j[\mathit{posiq}]
\end{align*}
where the $2/(n(i+2))$ term is computed at position $i$ by using future masked attention with a score of $2/n$ (computed earlier) at the first position and $0$ elsewhere.
This gives an attention score of:
\begin{equation}
    \label{eqn:attn-score-prefix-sum-invsq}
    S(i, j) = \frac{2}{n(j+2)(i+2)} - \frac{p_i + 2}{n(j+2)^2(i+2)}
\end{equation}
which, by \cref{thm:quadratic_maximization_alt}, is uniquely maximized when $j = p_i$. This allows us to retrieve value $j/n = p_i / n$ from position $j$, as needed in the proof in \cref{sec:prefix-sum}. %

The third and final use of $\mathit{posq}$ is in the simulation of \emph{leftmost and rightmost attention} and its \emph{default values} (\cref{sec:leftmost-rightmost-attention}). Specifically, suppose the attention predicate in \SRASP\ is an equality comparison of two integer vectors, say $v_1(i)$ and $v_2(j)$, represented as $u_i[r] = k_i / n$ and $u_j[s] = k_j / n$, respectively. In this case, we first use two lookup operations (\cref{lemma:copy-from-basic}, updated for the inverse square position embedding) with the $\mathit{posi}$ and $\mathit{posiq}$ entries of the position embedding to copy inverses and inverse squares of the values in $v_2$ to coordinates $t$ and $z$ of the activation. As in the original proof, let $u_{-1}, u_0, \ldots, u_{n-1}$ denote the resulting activation sequence. We thus have $u_i[t] = 1/(k_i + 2)$ and $u_i[z] = 1/(k_i + 2)^2$. We then use the attention score function
\begin{equation}
    S(i, j) = \frac{2}{n} \, u_j[t] - \left( u_i[r] + \frac{2}{n} \right) \cdot u_j[z]
\end{equation}
a bilinear combination of $u_i$ and $u_j$ equivalent to:
\begin{equation}
    S(i, j) = \frac{2}{n (k_j + 2)} - \frac{k_i + 2}{n (k_j + 2)^2} .
\end{equation}
By \cref{thm:quadratic_maximization_alt}, $S(i, j)$ is uniquely maximized over values of $j \geq -1$ when $k_j = k_i$.

As in the original argument, there may be multiple matches and we thus need to break ties in favor of the leftmost or rightmost match. To this end, we observe that $S(i, j) = 1/(n (k_i + 2))$ when $k_j = k_i$, and compare this to the maximum value of $S(i, j)$ for $k_j \neq k_i$, which is $(k_i + 4) / (n (k_i + 3)^2)$, attained at $k_j = k_i + 1$. Thus, the gap between the attention score when $k_j = k_i$ versus the maximum possible when $k_j \neq k_i$ is $1 / (n (k_i + 2) (k_i + 3)^2)$. Since $k_i < n$, this is lower bounded by $1 / (n (n+1)(n+2)^2) > g(n)$ where $g(n) = 1/(20 n^4)$. As in the original argument, if we add or subtract from $S(i, j)$ values less than $g(n)$, no non-equality score can exceed the corresponding equality score. We achieve this by adding or subtracting the tie-breaking term $g(n) j/(2n) = j/(40 n^5)$; the reason for using this specific tie-breaking term will become apparent when we discuss default values below. This term is computable by first computing $1/n^4$ at position $i$ and then using dot product attention with $j/n$ in the position encoding of $j$. In order to compute $1/n^4$, we can attend uniformly with only the first position having value $1/n$ (the rest having value $0$) to obtain $1/n^2$, and repeat this process twice more to obtain $1/n^4$. This finishes the updates needed for the simulation of leftmost and rightmost attention.

We address \emph{default values} in a similar way as in the original proof.
When it involves an equality comparison of integer vectors and rightmost attention, we observe that with the tie-breaking term $g(n) j/(2n)$ discussed above, the gap between the matching attention score $1/(n (k_i + 2))$ and the maximum non-matching attention score for rightmost attention is at least $g(n) / 2$.
Hence, a default position value of $1/(n (k_i + 2)) - g(n) / 4$ is strictly between these two values.
Further, this default position value is computable at position $i$ by the same arguments as above.
We treat default values with leftmost attention analogously.
\end{proof}

\section{Comparison with other simulations}
\label{sec:other_approaches}

In the prefix sum operation (\ref{eq:prefix_sum_result}), the result at position $i$ is $s(i)/(i+1)$, where $s(i)$ is the prefix sum of $v(i)$.
The fact that the denominator of this expression varies with position is an obstacle to comparing or adding the values $s(i)$ and $s(j)$ at two different positions $i$ and $j$.
This problem is addressed by \citet{yao2023selfattention} and \citet{merrill2024expressive} using a non-standard layer normalization operation to produce a vector representation of the quantities, which allows them to be compared for equality using dot product attention.
\Citet{perez-etal-2021-turing} include $1/(i+1)$ in their position embedding to enable the comparison; however, they compute attention scores as $-|\langle Qu_i, Ku_j\rangle|$ in place of the standard dot-product.
The approach of the current paper is based on that of \citet{barcelo2023logical}, who show how average hard attention can be used to compute the prefix sum of a $0/1$ vector.

\end{document}